\documentclass[twoside,11pt]{article}

\usepackage{jmlr2e}

\usepackage{amsmath,amsfonts,mathtools}
  \usepackage{graphicx}
  \usepackage{graphics}

\usepackage{tikz}
\usepackage{color} 
\usepackage{pgfplots}
\usepackage{paralist}

\newcommand{\A}{\mathcal{A}}    % Action space
\newcommand{\E}{\mathbf{E}}     % Expected value
\newcommand{\p}[1]{\mathbf{P}\left\{#1\right\}}
\newcommand{\rFVI}{\hat{r}^\ast_{x_0}(K,A)}     % Real numbers
\newcommand{\R}{\mathbb{R}}     % Real numbers
\newcommand{\T}{\mathsf{T}}     % Mapping
\newcommand{\W}{W}              % Value function
\newcommand{\We}{\hat{W}}       % Approximate Value function
\newcommand{\w}{w}              % A function in the function class

\newcommand{\X}{\mathcal{X}}    % State space
\newcommand{\Tx}[2]{T_x\left(#1 \mid #2 \right)}% state transition kernel
\newcommand{\Txo}{T_x}% state transition kernel
\newcommand{\tx}[2]{t_x\left(#1 \mid #2 \right)}% state transition kernel
\newcommand{\txo}{t_x}% state transition kernel

\newcommand{\xknext}{x_{k+1}}        % next state
\newcommand{\xb}{x_k^{i}}
\newcommand{\xs}{x_{k+1}^{i,a,j}}

\newcommand{\xh}{x^{i}} % hold out
\newcommand{\xsh}{y^{i,a,j}}

\newcommand{\Nh}{\tilde{N}} % hold out
\newcommand{\Mh}{\tilde{M}} % hold out

\newcommand{\Ca}{\mathbf{A}}
\newcommand{\Cb}{\mathbf{B}}
\newcommand{\Cc}{\mathbf{C}}

\newcommand{\Te}{\hat{\mathsf{T}}}

\newcommand{\M}[1]{\mathcal{#1}}

\newcommand{\Eo}[2]{\mathbf{E}_{#1}\left[#2\right]}

\newcommand{\ind}{\mathbf{1}}

\newcommand{\pseudo}[1]{\mathsf{dim}_p\left(#1\right)}
 
\newcommand{\acc}{\mathbf{\Delta}}
\newcommand{\Pacc}{\delta_{\mathbf{\Delta}}}
\allowdisplaybreaks
\usepackage{newfloat}
\DeclareFloatingEnvironment[fileext=frm,placement={!ht},name=Algorithm]{algorithm}
\usepackage{caption}
\ShortHeadings{Quantitative Approximations for Reach-Avoid over General Markov Processes}{Haesaert, Babuska and Abate}
\firstpageno{1}

\begin{document} 
%\begin{frontmatter}
\title{Sampling-based Approximations with Quantitative Performance for the Probabilistic Reach-Avoid Problem over General Markov Processes}

\author{\name Sofie\ Haesaert 
 \email s.haesaert@tue.nl \\
       \addr Control Systems Group, Department of Electrical Engineering\\
       Eindhoven University of Technology\\
       Eindhoven, The Netherlands
       \AND
       \name Robert Babuska\email r.babuska@tudelft.nl \\
       \addr Delft Center of Systems and Control\\
       Delft University of Technology\\
       Delft, The Netherlands
       \AND
       \name Alessandro Abate 
       \email alessandro.abate@cs.ox.ac.uk \\
       \addr Department of Computer Science\\
       University of Oxford\\
       Oxford, United Kingdom
       }

\editor{ }

\maketitle

%\begin{keyword}
%General state-space processes,
%reach-avoid problem,
%dynamic programming,
%fitted value iteration,
%computational approximation with error bounds.
%\end{keyword}
%\begin{keyword}
%reach-avoid problem, dynamic programming, policy synthesis, statistical learning theory, fitted value iteration 
%\end{keyword}
\begin{abstract}
This article deals with stochastic processes endowed with the Markov (memoryless) property and evolving over general (uncountable) state spaces.
The models further depend on a non-deterministic quantity in the form of a control input, which can be selected to affect the probabilistic dynamics. 
We address the computation of maximal reach-avoid specifications, together with the
synthesis of the corresponding optimal controllers.  
 % that maximize the probability associated to a rather general property of interest,  known as the ``reach-avoid'' specification. 
The reach-avoid specification deals with assessing the likelihood that any finite-horizon trajectory of the model enters a given goal set,
while avoiding a given set of undesired states. 
% (both arbitrary subsets of the state space). 
%Equivalently, the property can be expressed as entering a goal set, while dwelling within a set of allowed states. 
%The reach-avoid property is a well known specification that lies at the core of a number of modal logics used in the field of formal verification,
%and which is relevant for safety-critical applications ranging from robotics to air traffic management.
This article newly provides an approximate computational scheme for the reach-avoid specification based on the Fitted Value Iteration algorithm,
which hinges on random sample extractions, 
and gives a-priori computable formal probabilistic bounds on the error made by the approximation algorithm: 
as such, the output of the numerical scheme is quantitatively assessed and thus meaningful for safety-critical applications. 
% of the property of interest.  
Furthermore, 
we provide tighter probabilistic error bounds that are sample-based. 
% and can be computed given the outcome of the Fitted Value Iteration}.   
The overall computational scheme is put in relationship with alternative approximation algorithms in the literature, 
and finally its performance is practically assessed over a benchmark case study. 
\end{abstract}

%\end{frontmatter}
\section{Introduction}
\label{Introduction}  
This contribution concerns a problem grounded in concepts from a few different areas:
we deal with probabilistic processes evolving over continuous (and in particular uncountable) state spaces -- this leads to the use of measure-theoretical material from Stochastic Processes and Probability Theory \citep{MTH1993};
we work with models endowed with a control input and investigate control synthesis, which relate to a broad literature in Control Theory \citep{Bible}; 
%,k1967};
furthermore, we are interested in quantifying the probability associated to a dynamical property, known as reach-avoid,
which corresponds to a widely used model specification in the field of Formal Verification \citep{bk2008}; 
%,CGP99};
and finally we employ a sampling-based algorithm to approximately compute the likelihood associated to the above specification. 
The algorithm, 
 known as Fitted Value Iteration (FVI) \citep{Remi}, %which
 is a regression scheme developed in Machine Learning.\\ \\
We focus on stochastic processes endowed with the Markov property (where the future is independent of the past, conditional on the present) and,
aiming for generality, we deal with processes evolving over a continuous state space. 
%\removed{The presence of an uncountable state space requires dealing with measure-theoretical techniques, which adds in complication over %known 
%finite state models such as Markov chains %\citep{MTH1993}
%.} 
We are further interested in a class of such models known as stochastic hybrid systems (SHS) \citep{APLS08}, which are endowed with a ``hybrid'' (that is, both continuous and discrete) state space, which 
%. The results discussed in this work 
%can be exported to SHS, and 
are relevant for a number of applications in Engineering and the Life Sciences \citep{BL06,CL06}.
This work investigates the problem of controller synthesis over these models, 
namely the selection of sequences of control inputs (which in particular can be functions of the states of the model) over a finite time horizon,
in order to optimise a given figure of merit. 
%k1967 \citep{k1967}. 

As for the figure of merit of interest in this work, 
we choose to go beyond the classical properties investigated in Systems and Control theory, which by and large deal with known and standard problems of stability, regulation, and tracking.
Instead, 
we focus on the reach-avoid specification, a property that is well known and central within the Formal Verification field 
%, where rich dynamical properties are expressed via logical formulas or as automata
\citep{bk2008}. 
%,CGP99}. 
Notice that classical results in Formal Verification deal with simple models -- usually finite-state transition systems or Markov chains -- which allow for the development of 
%decidable and 
computational results, and which mostly deal with verification tasks that do not involve policy synthesis. In this work instead we consider reach-avoid specifications over models with continuous stochastic transitions and endowed with control inputs.%\\[-1.5em] 

The reach-avoid problem deals with computing the likelihood that, within a given finite time horizon, a trajectory of the model enters a goal set, 
while avoiding a given set of undesired states (both sets are arbitrary measurable subsets of the state space). 
Equivalently, the property can be expressed as the probability that the process enters a goal set while dwelling within a set of allowed states.
The reach-avoid property is a generalization of widely studied properties, 
such as reachability and invariance,
and represents a known specification (denoted as ``bounded until'') that lies at the core of a number of modal logics used in the field of formal verification,
such as Linear Temporal Logic and Computational Time Logic \citep{bk2008}. 
%,CGP99}.
From a controller synthesis perspective, the goal becomes that of either maximizing or minimizing the above specification 
%with the allowed control inputs and 
over the given time horizon. %\\[-1.5em]
%As an example for the reach-avoid property, 
%consider the model of an aircraft and seeking a path to maximize the likelihood of navigating toward a given target in the airspace, 
%while avoiding the presence of a storm. 

In the context of probabilistic models evolving over continuous domains and in discrete time (which is the framework considered in this work), 
the probabilistic reachability and reach-avoid specifications have been investigated in \citep{APLS08,Reachavoid}. 
These results have recently led to the study of other properties, either richer \citep{AbateQuanti} or defined over unbounded time horizons \citep{ta2014}. 
These results have focused on the theoretical \emph{characterization} of the specifications/properties of interest:
of course it is also of much interest to provide algorithms that can numerically compute these figures. 
Computational approaches to probabilistic reachability have been studied in \citep{Abate3,SA13}: 
the strength of these results is that the proposed numerical schemes have explicitly quantified error bounds.
This is unlike other, known approximation schemes in the literature \citep{KR06,KD01,PH06},
which provide results with properties that are only known asymptotically. %\\[-1.5em]
%\removed{Further, since the aforementioned formal approximation approaches are based on discretization techniques that result in finite-space models,
%they are also promising as automatic software tools for verification and synthesis, known in Formal Verification as model checkers \citep{HKNP06,KKZ05}.
%%\textcolor{red}{[cite How to abstract intelligence, If verification is in order. S. Pathak, L. Pulina, G. Metta, and Armando Tacchella.]}
%}

This article provides a new approximate computational scheme for the reach-avoid specification based on the Fitted Value Iteration algorithm, 
which hinges on random sample extractions. 
% say that it allows for the formal probabilistic bounds on the error which are dependent only on properties of the FVI algorithm; including complexity of the function class and the number of sample and not on properties of the markov process or reach-avoid property. Additionally this work includes new empirical estimations of the error bounds which, extended with a confidence interval, give a formal probabilistic error bound that depends on the properties of the Markov process and reach-avoid property.
This work originally derives formal probabilistic bounds on the error made by the approximation algorithm. 
In order to do so, the FVI scheme is tailored to the characterization of the reach-avoid problem, 
which leads to Dynamic Programming (DP) recursions based on a sum-multiplicative form that is non-standard since it departs from the classical additive (possibly discounted) cost functions \citep{Bible}. 
Starting from the regression bounds 
%on single step errors 
in \cite{Remi}, 
this work includes new results on the error 
%propagation 
for the FVI approximation and a-priori performance guarantees. 
%\textcolor{red}{[I would eliminate this sentence completely:]}
%Quite interestingly, 
%the provided probabilistic bounds are general and \textcolor{red}{distribution-free}\todo{Quite interestingly, the provided probabilistic bounds are based on general and distribution-free bounds.}. 
Additionally, 
novel and tighter probabilistic error bounds for dynamic programming solutions of the reach-avoid problem based on samples extraction are presented. 
%\removed{It is of interest to benchmark the outcomes of the proposed approach with those of}
As a comparison to the alternative techniques in the literature \citep{Abate3,SA13},  
we 
%\removed{will see} 
show the related techniques provide bounds that are valid deterministically, 
whereas the proposed result yields tighter results in general that are valid with a certain (tunable) confidence.  
%\todo{Robbert:{\tiny This does not sound particularly strong for our approach - are we not disqualifying it a bit? (from this sentence I gather that related techniques give "more valid" results - but I am probably overlooking something ...)}}
The outcomes lead to an approach providing controller synthesis with a certified performance, 
which is relevant for safety-critical applications \citep{BL06}. \\
%For the sake of space, longer proofs can be found at \cite{HBA14}.%\\[-1.5em]
%The article is structured as follows.
%Section \ref{SectionTwo} introduces the model and provides a theoretical characterization of the reach-avoid problem. \removed{-- it also puts forward the two alternative approaches to compute the quantities of interest.
%Among the two approaches, } In Section \ref{SectionThree} we introduce the FVI algorithm and tailor it to the reach-avoid problem under study, discussing in detail its implementation.
%Probabilistic bounds on the 
%%(probabilistic) 
%error of the FVI scheme versus the theoretical characterization of the reach-avoid problem are discussed in Section \ref{SectionFour}. 
%The a-priori, computable probabilistic bounds of Section \ref{SectionFour} are complemented with a-posteriori, sample-based probabilistic bounds in Section \ref{EmpError}. 
%Section \ref{SectionFive} details the application of the FVI algorithm to a case study, and Section \ref{Conclusions} concludes the work.
 The proofs of the statements are included in the Appendix.

\section{Probabilistic Reach-Avoid Problem over General Markov Processes}
\label{SectionTwo}\mbox{ }%\\[-4em] 
\begin{definition}[General Markov process]
A discrete-time general Markov process is comprised of:%\\[-2em]
\begin{itemize}
\item A continuous (uncountable) state space $\X \subset \R^n$;
\item An action space $\A=\{a_1,\ldots,a_m\}$ consisting of a finite number of actions;
\item A Borel-measurable stochastic kernel $ \Txo$, which assigns to each state-action pair $x\in \X$ and $a\in \A$ a probability distribution $\Tx{\cdot}{x,a}$ over $\X$.%\\[-1.4em]
\end{itemize}
\end{definition}
We denote with $(\X, \mathcal B(\X), P)$ a probability structure on $\X$, where $\mathcal B(\X)$ is the $\sigma$-algebra associated to $\X$ and $P$ is characterized as $P(y\in\!\! A | x\!\in\! \X, a\! \in \!\A)\! =\!  \int_A T_x(dy|x,a)$.
We assume that the stochastic kernels admit densities so that $\int_A \Tx{dy}{x,a} = \cramped{\int_A \tx{y}{x,a}dy}$.
%The notion of a Markov policy is given as follows.
%\\[-1.5em]
%
\begin{definition}[Markov policy]\label{def:MarkovPolicy}  A Markov policy $\mu$ over horizon $\cramped{[0,N_t]}$  is a sequence  \(\mu=(\mu_0, \mu_1, \ldots,\mu_{N_t-1})\) of universally measurable maps,
$
\mu_k:\X \rightarrow \A,   k=0,1,\ldots,N_t-1,
$ 
from the  state space $\M{X}$ to the action space $\mathcal{A}$.
The set of Markov policies is denoted as $\M{M}$. %\\[-1.4em]
%, and $\mu \in \M{M}$ has individual elements $\mu_k \in \M{M}_\mu$.
\end{definition}
The evolution of the general Markov process is considered over a finite horizon $k=0,1,\ldots,N_t$,
with $N_t\in\mathbb{N}$.
Consider a discrete-time general Markov process,
a Markov policy $\mu$, a deterministic initial state $x_0\in\X$
and a finite time horizon $N_t$:
an execution of the process characterizes a state trajectory given as $\{x_k|k=0,1,\ldots,N_t\}$.
The process evolves over the product space $\cramped{(\X)^{N_t+1}}$,
which is again endowed with a (product) $\sigma$-algebra and allows computing probability associated to events over trajectories --
we denote this probability by $\cramped{\mathbf{P}}$, and further define the probabilities $\cramped{\mathbf{P}_{x_0}, \mathbf{P}_{x_0}^{\mu}}$ as $\mathbf{P}$ conditioned on an initial state
and %respectively
 on an initial state and a policy, respectively.
The state at the $(k+1)$-st time instant, $x_{k+1}$,
is obtained as a realization of the controlled Borel-measurable stochastic kernel $\Tx{\cdot}{x_{k}, \mu_{k}(x_k)}$.
%As a trivial generalization, 
The model can be initialized according to an initial probability measure $P_0 \in M(\X)$, 
% \removed{(cf. Remarks \ref{rem:inprob}, \ref{rem:inprob2})},
where $M(\cdot)$ denotes the collection of probability distributions over a given set.
%\\[-2em]
%%%%%%%%%%%%%%%%%%%%%%%%
\subsection{Probabilistic Reach-Avoid Problem: Definition}  \mbox{ }%\\[-4em]
%Definition and Characterization}

Let us define the probabilistic reach-avoid problem,
also known as constrained reachability \citep{bk2008},
and provide its characterization.
Consider a 
%finite-horizon reach-avoid property over a 
safe set $A\in \M{B}(\X)$,
a target set $K \in \M{B}(\X)$,
and a finite time horizon $N_t\in \mathbb{N}$. 
A given state trajectory $\{x_k|k=0,1,\ldots,N_t\}$ verifies the reach-avoid property if it reaches the target set $K$ within the time horizon,
while staying inside the safe set $A$. % beforehand.
This property can be expressed as%\\[-1.4em]
\begin{align*}
\exists j \in [0,N_t]: x_j \in K \wedge \forall i \in [0,j-1]:\, x_i \in A\setminus K.\end{align*}%\\[-1.5em]
%for the given state trajectory $\{x_k|k=0,1,\ldots,N_t\}$.
Let us now consider the probabilistic reach-avoid property for a general stochastic system,
%Instead of looking at the reach-avoid property for a given state trajectory,
defined as the probability that
%it holds for a random execution of the stochastic system. In other words we look at the probability that
an execution associated with a fixed Markov policy $\mu \in \M{M}$ and an initial condition $x_0\in\X$ reaches the target set $K$ while avoiding $\M{X}\setminus A$. 
Formally, %\\[--1.4em]
%This is formalized as
\begin{align}\label{eq:logical_RA}
r_{x_0}^\mu(K,A) = \mathbf{P}_{x_0}^\mu \left\{ \exists j \in [0,N_t]: x_j \in K
%\right.\qquad\qquad\notag\\\left.
\wedge \; \forall i \in [0,j-1]: x_i \in A\setminus K\right\},
\end{align}%\\[-1.5em]
where the 
%consecutive 
states $x_0,x_1,\ldots,x_{N_t}\in \M{X}$ are sampled via
%, and the transition between these states are governed by the Borel measurable
the stochastic kernel $\Txo$ under policy $\mu$.
The formula contained in (\ref{eq:logical_RA}) can be written as a boolean expression % by use of 
using indicator functions,
which leads to 
%its expression as 
an expectation over the state trajectories as%\\[-1.4em]
\begin{align*}\textstyle
r_{x_0}^\mu(K,A)=\E_{x_0}^{\mu}\bigg[\sum\limits_{j\in[0,N_t]}\bigg(\prod_{i=0}^{j-1}\ind_{A\setminus K}(x_i)\bigg)\ind_{K}(x_j)\bigg],
\end{align*}%\\[-1.2em]
where $\ind_{B}(x) = 1$ if $x \in B$, else it is equal to $0$.
The reach-avoid problem subsumes other known problems widely studied in System and Control Theory and in Formal Verification,
such as that of reachability of set $K$, which is simply obtained by selecting $A = \X$,
or that of invariance within a set $B$, which is characterized as the dual of the reachability problem over set $\X \setminus B$. 

For a given
%property and
policy $\mu$, the time-dependent value function $W_k:\X\rightarrow [0,1]$, defined as%\\[-1.8em]
\begin{align*}\textstyle \W_k^\mu(x)\!=\E^\mu\!\bigg[ \sum\limits_{j\in [k+1,N_t]}\!\!\!\bigg( \prod\limits_{i=k+1}^{j-1}\!\!\!\ind_{A\setminus K}(x_i)\bigg)\ind_K(x_j)\bigg|x_k=x\bigg],
\end{align*}%\\[-1.2em]
is the probability that the state trajectory $\{x_{k+1},\ldots, x_{N_t}\}$,
starting from $x_k$, will reach the target set $K$ within the time horizon $[k,N_t]$,
while staying within the safe set $A$.
This function allows expressing the reach-avoid probability backward recursively,
as follows. 
% the following result elucidates.
%
\begin{proposition}\label{prop:DP}
Given a policy $\mu=(\mu_0, \mu_1, \ldots,\mu_{N_t-1})$,
define function $\W_{k}^\mu : \X\rightarrow [0,1]$ by backward recursion%\\[-1.4em]
\begin{align*} \textstyle
\W_{k}^\mu(x)=\E_{x}^{\mu_k}\left[\ind_{K}(\xknext)+\ind_{A\setminus K}(\xknext)\W_{k+1}^\mu(\xknext)
\right],
\end{align*}%\\[-1.8em] with  $ \xknext\sim \Tx{\cdot}{x,\mu_k(x)}$ for $k=N_t-1,N_t-2,\ldots,0 $, 
and initialized with $\W_{N_t}^\mu(x)=0$. 
Then for any initial state $x_0\in \X$, the probabilistic reach-avoid property $r_{x_0}^\mu(K,A)$ can be expressed as %\\[-1.4em]
\begin{align*}
r_{x_0}^\mu(K,A)=\ind_{K}(x_0)+\ind_{A\setminus K}(x_0) \W_0^\mu(x_0)\ .\end{align*}
\end{proposition}\mbox{ }%\\[-2.5em]
\begin{proof}
The proof follows \cite[Lemma 4]{Reachavoid},
%which is based on a scheme in \citep{APLS08} and 
where the above statement is proven for a value function $V_k^\mu(x)=\ind_K(x)+\ind_{A\setminus K}(x)W_k^{\mu}(x)$. 
\end{proof}\\*[1em]
Notice that, while the probabilistic reach-avoid problem has been formulated above via DP recursions,
it hinges on a sum-multiplicative characterization which is non-standard:
much of the analytical and computational results in DP are %indeed 
formulated for additive (possibly discounted) cost functions \citep{Bible}.

Rather than selecting and fixing a policy $\mu$ as done above,
we now focus on the controller synthesis problem,
which seeks the Markov policy $\mu^\ast$ that maximizes the probabilistic reach-avoid property,
and which is such that $r_{x_0}^{\ast}(K,A) = \sup_{\mu\in \M{M}}r^\mu_{x_0}(K,A)$. 
%i.e.,
%Thus the optimal reach-avoid property can be expressed as an expectation over the state trajectories for the optimal policy
%\begin{align*}
%&r_{x_0}^{\ast}(K,A)
%%= \sup_{\mu\in \M{M}}r^\mu_{x_0}(K,A)
%%\\&
%= \!\!\sup_{\mu \in \M{M}}\E_{x_0}^{\mu}\bigg[\sum_{j\in[0,N_t]}\bigg(\prod_{i=0}^{j-1}\ind_{A\setminus K}(x_i)\bigg)\ind_{K}(x_j)\bigg].\end{align*}%\\[-1.2em]%For a given property, the optimal time-dependent value function $W^\ast_k:\X\rightarrow [0,1]$,
%\[
%\W^\ast_k(x)=\sup_{\mu\in \M{M}}\E^{\mu}\left[\left.\sum_{j\in [k+1,N_t]}\left(\prod_{i=k+1}^{j-1}\ind_{A\setminus K}(x_i)\right)\ind_K(x_j)\right|x_k=x\right],
%\]
%is the maximum probability that the state trajectory $\{x_{k+1},\ldots, x_{N_t}\}$ starting from $x_k$ will reach the target set $K$ within the time horizon $[k,N_t]$ while staying in the safe set.
Let us emphasize that the optimization is over finite-action policies, 
which are however functions of the continuous state space. 
The optimal policy can be characterized as follows.
\begin{proposition}\label{prop:optimalDP}
Define functions $\W_{k}^{\ast} : \X\rightarrow [0,1]$, by the backward recursions%\\[-1.3em]
\begin{align*}%
\W_{k}^{\ast}(x)=\max_{a\in \A}\E_{x}^a\left[\ind_{K}(\xknext)+\ind_{A\setminus K}(\xknext)\W_{k+1}^\ast(\xknext)\right],
%\right.\\ \left. 
%\mid \xknext\sim \Tx{\cdot}{x,a}
\end{align*}%\\[-1.3em]
with $\xknext\sim \Tx{\cdot}{x,a}$ 
for $k=N_t-1,N_t-2,\ldots,0$, and initialized by $W_{N_t}^\ast(x)=0$.
Then for any initial state $x_0\in \X$ the optimal probabilistic reach-avoid property $r_{x_0}^\ast(K,A)$ can be expressed as%\\[-1.3em]
\begin{align*} 
&r_{x_0}^\ast(K,A)=\ind_{K}(x_0)+\ind_{A\setminus K}(x_0) \W_0^{\ast}(x_0).\hspace{2cm} 
\end{align*}%\\[-1.3em]
Furthermore, $\mu_k^\ast:\X\rightarrow \A$ for $k=N_t-1,N_t-2,\ldots,0 $, is such that $\forall x \in \X$:%\\[-1.3em]
\begin{align*} 
\mu_k^\ast(x)=\arg\max_{\mathclap{a\in\A}} \E_x^a\!\left[\ind_{K}(\xknext)
%\right.\right.\\& \left. \left.
\! +\!\ind_{A\setminus K}(\xknext)\W_{\cramped[\scriptstyle]{k+1}}^\ast(\xknext)\right] %\mid \xknext\sim \Tx{\cdot}{x,a}
\end{align*}%\\[-1.4em]
and $\mu^\ast=(\mu_0^\ast, \mu_1^\ast, \ldots,\mu_{N_t-1}^\ast)$ is the optimal probabilistic reach-avoid Markov policy.
%A sufficient condition for the existence of $\mu^\ast$ is that
%\[
%U_k(x,\lambda)=\left\{a\in \A : \E\left[\ind_{K}(\xknext)+\ind_{A\setminus K}(\xknext)\W_{k+1}^\ast(\xknext)|x,a\right] \geq\lambda\right\}
%\]
%is compact $\forall x \in \X$, $\lambda\in \R$, and for $k=N_t-1,N_t-2,\ldots,0 $.
\end{proposition}
\begin{proof}
See again \cite[Theorem 6]{Reachavoid} and \citep{APLS08}.
%where the above proposition was proved for a value function $V_k^\ast(x)=\ind_K(x)+\ind_{A\setminus K}(x)W_k^{\ast}(x)$.
\end{proof}

For a given time horizon $N_t$,
the computation of  $r_{x_0}^\ast(K,A)$,
as in Proposition \ref{prop:optimalDP},
can be seen as the application of $N_t$ mappings.
More precisely, let us define a dynamic programming operator $\T$ as 
%%\\[-1.4em]
%\begin{align*} 
$\W^\ast_k=\T \W_{k+1}^\ast$,  
%\end{align*}
%%\\[-1.4em]% value iteration?
such that for all states $x\in\X$, the function $\W_{k}^\ast:\X\rightarrow[0,1]$ is defined as%\\[-1.4em]
\begin{align}\label{eq:T} 
\W_{k}^\ast (x) &= \left(\T \W^\ast_{k+1}\right)(x)\\&
=\max_{a \in \M{A}}\E_x^a\!\!\left[\ind_K(\xknext)\!+\!\ind_{A\setminus K}(\xknext)\W^\ast_{k+1}(\xknext)
\right].\notag%\\&   \qquad\qquad\qquad\left.
%\big{|}\xknext\sim \Tx{\cdot}{x,a}\right].
\end{align}%\\[-1.4em]
The value of the optimal probabilistic reach-avoid property can be written as the composition of $N_t$ mappings,%\\[-1.2em]
\begin{align*} r_{x_0}^\ast(K,A)=\ind_{K}(x_0)+\ind_{A\setminus K}(x_0)\left(\T^{N_t}W^\ast_{N_t}\right)(x_0).%%\\[-%\\[-2em]
\end{align*}%

%\\*[-1.2em] with $W^\ast_{N_t}(x)=0$.
%\begin{cor}[Initial probability distribution]\label{rem:inprob}
%It is often relevant to compute the probability of the reach-avoid problem with reference to an initial distribution, rather than a single state.
%Given an initial state $x_0 \sim P_0 \in M(\X)$,
%the optimal probability that traces of the Markov model reach the target set $K$ within $N_t$ time steps, while staying in $A$, is obtained as
%\begin{align*}
%r_{P\mathclap{_{_0}}}^\ast (K,A) =\!\! \int_{\mathclap{\ \X}} \ind_K(x_0)\!+\!\ind_{A\setminus K}(x_0)\left(\T^{N_t}W_{N_t}\right)\!\!(x_0)P_0(dx_0).
%\end{align*}%
%%A straightforward adaptation of the last iteration of step (3.) in Algorithm \ref{alg:FVI} yields an approximation of this quantity.
%%Let the algorithm do a single additional iteration and obtain the function $\We^\ast_{0}$. Then estimate the above integral with a Monte-Carlo integration, where is used $\We^\ast_{0}$ to approximate the function $\T^{N_t}W_{N_t}$.
%\end{cor}
%%%%%%%%%%%%%%%%%%%%%%%%
\subsection{Computation of the Reach-Avoid Probability}%
Notice that generally it is not possible to solve the above recursions exactly:
in order to determine the backwards iteration at a single point $x_i\in \X$,
namely $\cramped{\W_{k}^\ast(x_i)}=\big(\T \W_{k+1}^\ast\big)(x_i)$,
one should exactly solve (\ref{eq:T}).
The exact solution of (\ref{eq:T}) however is seldom analytical and can possibly result in computationally expensive procedures.
%
%And while a point-wise representation is infeasible for the continuous state space, an
%
The absence of an analytical representation for $\cramped{W^\ast_k}:\X\rightarrow [0,1]$ leads to the use of approximation techniques,
which can be categorized in two families: %\\[-1.4em]
\begin{compactenum} \item \emph{Numerical approximation techniques}, which provide an approximation of the optimal probabilistic reach-avoid problem with actual error bounds.
More precisely, 
%\removed{ suppose that $r_{x_0}^\ast(K,A)$ is the quantity that in the end we want to approximately compute.}
for a given error bound $\acc>0$ we seek a numerical scheme that obtains an approximation $\hat{r}^\ast_{x_0}(K,A)$,
which is such that $|\hat{r}^\ast_{x_0}(K,A)-r_{x_0}^\ast(K,A)|\leq \acc$.
This approach is taken in 
%Abate5,
\citep{Abate3,SA13},
and the scheme is prone to suffer from the curse of dimensionality,
since it approximates a general stochastic system with a Markov chain by partitioning the state space.
\item \emph{Probabilistic approximation techniques}, which approximate the original problem with probabilistic guarantees.
The obtained approximation scheme $\hat{r}^\ast_{x_0}(K,A)$ for $r_{x_0}^\ast(K,A)$ depends on a finite number of samples or sampled paths of the underlying model. 
%\removed{Notice that there exist no such schemes in the literature, while there exist many results with bounds that are known to converge asymptotically as the number of sample grows
%\citep{KD01}.}
%Therefore any approximation $\hat{r}^\ast_{x_0}(K,A)$ of $r_{x_0}^\ast(K,A)$ is in fact a realization of a random variable.
For a given error bound $\acc>0$ and confidence $1-\emph{}\Pacc$,
the probability that the approximation is not close to the the optimal value can be bounded probabilistically as%\\*[-2.5em]
\begin{align}\label{eq:acc} \p{\left|\hat{r}^\ast_{x_0}(K,A)-r_{x_0}^{\ast}(K,A)\right|>\acc}\leq \Pacc. \end{align}%
\end{compactenum}
%The numerical methods  \citep{Abate1,Abate2,Abate3,Abate5,Reachavoid} often have unrealistic demands on computational power and memory usage for higher dimensions of the state space caused by bad scaling behavior.
In this work, we newly pursue the second approach 
%sketched in the second point 
by focusing on results from the area of learning, and in particular on algorithms for functional approximations. 
A learning approach is 
%\removed{in particular considered} 
suitable for complex systems,
such as general Markov processes,
since it replaces model-based evaluations by model-free, sample-based evaluations \citep{RL_Functionapprox}. 
We adopt the Fitted Value Iteration scheme (FVI) \citep{Remi}, 
 a learning algorithm fit in particular for finite horizon settings.\\*
%
 % --> section 3 FVI
 % --> section 4 probabilistic error bounds 
In practice, 
%probabilistic 
bounds for probabilistic approximation methods (\ref{eq:acc}) can be divided into two groups. 
Firstly,   
general model-free and  
sample-free bounds, 
which provide \emph{a-priori guarantees} on the achievable accuracy for a finite sample set. 
Though they can show convergence in probability up to a bias term, 
their generality can render them conservative when used as a tool to assign an accuracy guarantee. 
Alternatively, 
model-based and sample-based bounds: 
these bounds verify the accuracy of a dynamic programming scheme by drawing samples of the model and using available information from the model, 
and from the specific reach-avoid property under study.   
In the analysis of the algorithm these bounds can be perceived as complementary. 
%Therefore first 
A-priori and sample-free
bounds are derived in Section \ref{SectionFour} based on model-free/distribution-free notions, 
whereas model-based and sample-based bounds are given in Section \ref{EmpError}.

\section{Fitted Value Iteration}
\label{SectionThree}
% !TEX root = ALearningApproach.tex
% fitted value iteration

%-  estimate value function
%-  regression over value function
%-  set of observations/samples
%   => multiple sample batch

In this section we consider a learning algorithm that has been developed to solve 
%(possibly discounted) 
additive-cost optimal control problems,
and adapt it to the reach-avoid optimal control setting. 
Known as 
%fitted value iteration  
FVI \cite{Remi},
the algorithm 
%is based on 
extracts a finite number %numbers
of samples from the underlying model to numerically approximate the value recursions in (\ref{eq:T}). 
%\removed{,which we argue seldom admits analytical solutions}.
More precisely, the scheme generalizes the information gathered from the samples to approximate the ``exact'' optimal value function $\W^\ast_k$ as $\We^\ast_k$ in two steps:
first by estimating $\W^\ast_k$ over a finite number of states,
thereafter fitting the analytical function $\We^\ast_k\in\M{W}$ to the estimate. 
We define $\M{W}$ to be a strict subset of $B(\X;1)$,
the class of measurable functions defined over $\X$, lower bounded by 0 and upper bounded by 1.

%The FVI algorithm employs samples that are generated by the underlying model.
%At each iteration ($k = N_t-1,\ldots,0$) the data set comprises:
%\begin{enumerate}
%\item
%$N$ base points  $(\xb)_{1\leq i \leq N}$,
%each of which is drawn from the distribution $\eta$ over the set $A\setminus K$;
%\item
%$M$ samples for each base point $\xb\in\X$ and action $a\in\A$,
%which are drawn from the transition kernel as $\xs \sim \Tx{\cdot}{\xb,a}$.
%This set of samples is denoted by $\cramped{ (\xs )_{1\leq j\leq M}}$.
%\end{enumerate}
%%
\begin{figure}[h]\begin{minipage}{\linewidth}
{\rule{.975\linewidth}{1.5pt}}%\\[-1.75em]
\captionof{algorithm}{Sample generation for the FVI algorithm}\label{alg:samples}
{\rule{.975\linewidth}{1.5pt}}
{\begin{minipage}[h]{.95\linewidth}
 %\vspace{0.05cm}
 \small \it
% \vspace{0.05cm}
% \textbf{Sample generation}\\
 Given a safe set $A$, a target set $K$, a time horizon $N_t$, and distribution $\eta$,
 generate $N$ base points, and $M$ samples at each base point as follows:
  \begin{compactenum}
  \item
  Draw $N$ base points $\cramped{\left( \xb\right)_{1\leq i\leq N}}$ from the distribution $\eta$ in $A\setminus K$;
  \item
  Draw $M$ samples at each base point $\xb$ and at each actions $a\in \A$ from the stochastic kernel $\Txo$,
  and denote the set of samples as $\left(\xs\right)_{1\leq j\leq M}$.
  \end{compactenum}
 \end{minipage}}
{\rule{.975\linewidth}{1.5pt}}
\end{minipage}
\end{figure} 
The FVI algorithm employs samples that are generated by the underlying model. 
%
%The sample sets consist of 
Samples referred to as ``base points'' are taken from a chosen distribution $\eta$, 
and additional samples are drawn at the base points from the transition kernel. 
Algorithm \ref{alg:samples} summarizes the sample generation. 
Let us remark that at each iteration $k = N_t-1, \ldots, 0$,
a new set of samples is generated and used.
We denote as ``sample complexity'' the cardinality of the sample generation, namely $N$ and $M$. 
% \new{(cf Algorithm \ref{alg:samples})}.

At each (backward) iteration $k = N_t-1,\ldots,0$,
the algorithm executes two steps in order to approximate the exact recursion in (\ref{eq:T}):
\begin{enumerate}
\item
The first step consists of estimating the value of the backward mapping $ (\T \We^\ast_{k+1} )(\xb)$ at $N$ base points $\xb$.
The recursion  $(\T \We^\ast_{k+1} )(\xb)$ is estimated by an empirical operator $ \Te $ as follows:
\begin{align}\label{eq:FVIstep1}
&\hspace{-.5cm}\big(\Te \We^\ast_{k+1}\big)(\xb)=\max_{a \in \M{A}}\frac{1}{M}\sum_{j=1}^{M}\ind_{K}(\xs)
%\hspace{.8cm}\\&\hspace{2cm}
+\ind_{A\setminus K}(\xs)\We^\ast_{k+1}(\xs).%\notag
\end{align}%\\[-1.2em]
Here $\cramped{\xs}$ represent $M$ independent and identically distributed realizations obtained from $\Tx{\cdot}{\xb,a}$.
Hence, $j \in [1, M], i \in [1, N]$, and $a \in \mathcal{A}$.
For an increasing number of samples $\xs$,
the estimate $ (\Te \We^\ast_{k+1} )(\xb)$ converges to $ (\T \We^\ast_{k+1} )(\xb)$ with probability 1, by the law of large numbers.

\item
In the second step, function $\cramped{\We^\ast_{k}\in\M{W}}$
%(function class $\M{W}$ is discussed in Sec. \ref{sec:empnorm})
is estimated as the solution of
\begin{align}\label{eq:FVIstep2}
\We^\ast_{k}=\arg\min_{\w\in\M{W}}  \sum_{i=1}^{N} \left|\w(\xb)-\Te \We^\ast_{k+1}(\xb) \right|^p.\end{align}%\\[-1.2em]
%\removed{where the function class $\M{W}$ is a subset of the class of measurable functions that are bounded within the interval $[0,1]$.}  
We assume that the argument of the minimum belongs to the function class (this fact will be further discussed below).
The base points $\cramped{(\xb)_{1\leq i \leq N}}$ are independently drawn from a distribution $\eta$ supported over the set $A\setminus K$.
The power factor $p\geq1$ is a given positive number.
Given a function $\w(x)$ and an increasing value of $N$, the %term
summation in (\ref{eq:FVIstep2}) converges to  $\int_{A\setminus K} |\w(x)-\Te\We^\ast_{k+1}(x)|^p \eta(x) dx$, by the law of large numbers. 
This leads (cf. Section \ref{SectionFour}) to the convergence
%\footnote{In reality convergence can only be proved for certain function classes, cf. Section \ref{SectionFour}.}
of the argument $\We^\ast_k$ to the optimal fit that minimizes the distance between $ \Te \We^\ast_{k+1} $
%$\We^\ast_k$
and the functions $w \in \M{W}$
%is minimized
with respect to the $p$-norm,
weighted by a distribution with density $\eta$ and supported over the set $A\setminus K$, namely%\\[-2em]
\noindent{\small\begin{align*}
  \hspace{-1.5em}\|\We^\ast_{k}\!-\!\Te\We^\ast_{k+1} \|_{p,\eta}=\! \bigg(\!\int_{A \setminus K}\!\!|\We^\ast_{k}(x)-\!\Te\We^\ast_{k+1}(x)|^p \eta(x)dx\bigg)^{\!\!\frac{1}{p}}.\end{align*}}
\end{enumerate}

The overall FVI algorithm is summarized in Algorithm \ref{alg:FVI}.
The iterations are initialized as $\We^\ast_{N_t}(x)=0, x \in \X$,
and updated over the functions $\We^\ast_{N_t-1},\We^\ast_{N_t-2},\ldots,\We_1^\ast$.
Finally the value function at $k=0$ is approximated at the initial condition $x_0$ with a sample-based integration,
similar to (\ref{eq:FVIstep1}), using $M_0$ independent and identically distributed realizations of $\Tx{\cdot}{x_0,a}$ for all $a\in\A$.

\begin{figure}[htp]\begin{minipage}{\linewidth}
{\rule{.975\linewidth}{1.5pt}}%\\[-1.75em]
\captionof{algorithm}{Fitted Value Iteration algorithm}\label{alg:FVI}
{\rule{.975\linewidth}{1.5pt}}
{\begin{minipage}[h]{.95\linewidth}
 \vspace{0.05cm}
 \small \it
 Given an initial condition $x_0\in\X$,
 a safe set $A$,
 a target set $K$,
 a time horizon $N_t$,
 a set of $N$ base points and of $M$ samples at each base point (for each iteration $k$),
 a number $p$ and a distribution $\eta$,
 perform: 
% proceed with the following value iteration:
%\vspace{0.05cm}
%  \textbf{Value iteration}
    \begin{compactenum}
    \item Initialize $\We^\ast_{N_t}(x)=0, \forall x\in \X$;
    \item For $k=N_t-1$ to $1$ do
        \begin{compactenum}
        \item Collect samples (cfr. Algorithm \ref{alg:samples});
        \item Estimate $\T \We^{\ast}_{k+1}$ as $\big(\Te \We^\ast_{k+1}\big)(\xb)$ (cfr. \eqref{eq:FVIstep1})
%{  \[\hspace{-1cm} \begin{aligned}\left(\Te \We^\ast_{k+1}\right)(\xb)=\max_{a \in \M{A}}\frac{1}{M}\sum_{j=1}^{M}\left[ \mathbf{1}_K(\xs)\qquad\qquad\right.\\ \left.+ \mathbf{1}_{A\setminus K}(\xs) \We^\ast_{k+1}(\xs)\right];\end{aligned} \]}
        \item Find the function that minimizes the empirical $p$-norm as
         \begin{align*}\hspace{-1cm}\We^\ast_{k}=\arg \min_{\w \in \M{W}}\sum_{i=1}^{N} \left|\w(\xb)-\Te \We^\ast_{k+1}(\xb) \right|^p;\end{align*}%\\[-2.2em]
         \end{compactenum}
%         $\quad \quad\}$
            \item Collect $M_0$ samples for the single initial condition $x_0$ and for every action $a\in\A$, and estimate $\We^\ast_0(x_0)$ as in step (2)(b);
%    \item Estimate $\We^\ast_0(x_0)$ with\\
%     $\We^\ast_{0}(x_0)=\max_{a \in \M{A}}\frac{1}{M_0}\sum_{j=1}^{M_0}\left[ \mathbf{1}_K(y_{0,a,j})+ \mathbf{1}_{A\setminus K}(y_{0,a,j}) \We^\ast_{1}(y_{0,a,j})\right].$
    \item
    Return reach-avoid probability
	\begin{align*}\rFVI=\mathbf{1}_{K}(x_0)+\mathbf{1}_{A\setminus K} (x_0) \We^\ast_{0}(x_0), \forall x_0 \in \X.\end{align*}%\\[-2em]
     \end{compactenum}
%      \textbf{Probability of the reach-avoid problem}
%\[\rFVI=\mathbf{1}_{K}(x_0)+\mathbf{1}_{A\setminus K} (x_0) \We^\ast_{0}(x_0),\quad \forall x_0 \in \X.\]\vskip-0.1cm
 \end{minipage}}\\[1em]
{\rule{.975\linewidth}{1.5pt}}
\end{minipage}
\end{figure}
%\begin{remark}[Initial Probability Distribution]\label{rem:inprob2} \mbox{ }\\*
%In order to approximate the computation of the probability of the reach-avoid problem with reference to an initial distribution
%(rather than a single state, as discussed in Remark \ref{rem:inprob}),
%a straightforward adaptation of the last iteration of step (3.) in the algorithm of Table \ref{alg:FVI} is sufficient.
%\end{remark}

\begin{remark}[%
%Estimation 
%Synthesizing 
Approximately Optimal Policy]\label{rem:pol}
The FVI algorithm can be extended to include the 
%estimation 
synthesis of a policy $a=\hat{\mu}_{k}(x)$. At every iteration in Algorithm \ref{alg:FVI}, first the policy is estimated at all the base points $\xb$, as the argument of (2)(b). 
Secondly a classification algorithm is used, providing an approximately optimal policy $\hat{\mu}_{k}:\X\rightarrow \A$ for each $k$.  
\end{remark}

\section{A-Priori Probabilistic Error Bounds}
\label{SectionFour}
Let us recall the accuracy of the FVI algorithm as in (\ref{eq:acc}):
we say that the FVI algorithm has an accuracy $\acc$, with a confidence $1-\Pacc$,
if the probability that the error made by the approximate solution $\rFVI$ is larger than $\acc$,
is upper-bounded by $\Pacc$.
%More formally, the accuracy is given as \ref{eq:acc}: \begin{equation}\label
%\p{\left|\rFVI-r_{x_0}^{\ast}(K,A)\right|>\acc}\leq \Pacc.
%\end{equation}
%where $1-\Pacc$ is the confidence associated to the bound $\acc$ on the solution of the FVI.
%In this section 
We explicitly quantify the accuracy in (\ref{eq:acc}) and analyse it in two steps:
%\begin{compactenum}
%\item
first by computing a bound on the error of a single iteration (Sec. \ref{ErrorIteration}); 
%\item
then by studying the propagation of the single-step error over multiple iterations (Sec. \ref{ErrorFVI}).
%\end{compactenum}
Notice that this section provides a-priori bounds which are model- and 
%sample 
distribution-free, 
hence computable before applying the FVI algorithm. 
Alternatively, 
a-posteriori error bounds based on an additional sample set are proposed in Section \ref{EmpError}. 

\subsection{Error Bounds on a Single Iteration of the FVI}
\label{ErrorIteration}
% The accuracy of each backward mapping  \(\| \We^\ast_{k} - \T \We^\ast_{k+1}\|_{p,\eta}\) depends on the accuracy of the first step and second step of the FVI iterations.
%\removed{\red{\mbox{ }%\\[-2em](In this subsection the use of the time parameter $k+1$ aligns with the notations previously introduced, cfr. Equation (\ref{eq:FVIstep1}) and step 2.(b),
%%in Table \ref{alg:FVI},
%as well as (\ref{eq:FVIstep2}) and 2.(c) in Algorithm \ref{alg:FVI}.)}\\[.5em]}
Let $\T \We^\ast_{k+1}:\M{X}\rightarrow[0,1]$ be an unknown map,
and consider a function class $\M{W}\subset B(\X;1)$.
% consisting of functions supported on $\X$ and taking values in $[0,1]$.
Recall that at each iteration $k=N_t-1,N_t-2,...,1$, the objective of the learning algorithm is to find a function $w\in \M{W}$ that is close to $\T \We^\ast_{k+1}$ with respect to the following weighted, $p$-norm:
\begin{align}\label{eq:objective}
\|\w-\T\We_{k+1}^\ast\|_{p,\eta}=\bigg(\int_\X \left|\w-\T\We_{k+1}^\ast\right|^p \eta(x) dx\bigg)^{\mathclap{{\frac{1}{p}}}}.
\end{align}%\\[-1.2em]
%\textcolor{red}{Eliminate? For a given $\We_{k}^\ast$, the resulting single-step error can be expressed as  \( \|\We_{k}^\ast-\T\We_{k+1}^\ast\|_{p,\eta}.\)}
Notice that if $\T\We_{k+1}^\ast\not\in  \mathcal{\W}$, the optimal approximation
$
\inf_{\w\in \mathcal{\W}} \|\w-\T\We_{k+1}^\ast\|_{p,\eta}$
provides only a \emph{lower bound} on this error:
the presence of this non-zero bias error indicates that the FVI scheme is not asymptotically consistent,
namely that the error
%is independent of the number of samples, therefore the overall error of the FVI
does not converge to zero for an increasing sample size.
Of interest to this work,
a general \emph{upper bound} on $\inf_{\w\in \M{W}}  \|\w-\T \We^\ast_{k+1}\|_{p,\eta}$
%holding for any \textcolor{red}{$\We^\ast_{k+1}\in \M{W}$},
is derived as
\begin{equation}\label{eq:inherentBellman}
        d_{p,\eta}(\T \M{W}, \M{W})=\sup_{g\in \M{W}}\inf_{f \in \M{W}}\|f-\T g\|_{p,\eta}.
\end{equation}
In \citep{Remi} the bias $d_{p,\eta}(\T \M{W}, \M{W})$ is referred to as the \emph{inherent Bellman error} of the function space $\M{W}$.

As discussed, the FVI algorithm employs empirical estimates, given in \eqref{eq:FVIstep1}-\eqref{eq:FVIstep2}, of the quantity in \eqref{eq:objective}.
Therefore, the single step error hinges both on the inherent Bellman error,
and on the deviations caused by using estimates of the recursion step $\cramped{\T \We^\ast_{k+1}}$ over the base points (cfr. Section \ref{sec:estimationerror})
and of the norm $\|\cdot\|_{p,\eta}$ as the integral in \eqref{eq:objective} (cfr. Section \ref{sec:empnorm}).
%
%These are bounded by
%the error introduced by using an estimation of the recursions $\Te \We_{k+1}^\ast$ over the base points (Section \ref{sec:estimationerror})
%and by the maximal deviation of the empirical evaluation of the integral in \eqref{eq:objective} (Section \ref{sec:empnorm}).
%
The error contributions depend on the number of samples used ($N, M$) and on the capacity of the function class $\M{W}$ (Section \ref{sec:empnorm}, and Appendix \ref{proof:NumberSamples}),
whereas they do not depend on the distribution $\eta$, nor on the stochastic state transitions characterizing the model dynamics:
as such the bounds are general and ``distribution-free''. 

In the following subsections, two lemmas are derived, which are necessary to obtain a general upper bound for the single step error.
% as a function of the number of samples and capacity of the function class alone.
First (Section \ref{sec:estimationerror}), the error introduced by using the estimation $\Te \We_{k+1}^\ast$ of the recursion step is bounded using Hoeffding's inequality \citep{Hoeffding}.
Then in Section \ref{sec:empnorm} the maximal deviation of the empirical evaluation of the integral in \eqref{eq:objective} is bounded using methods from Statistical Learning Theory \citep{Vapnik}.

\subsubsection{Accuracy of the Estimation of $\cramped{\T \We^\ast_{k+1}}$}\label{sec:estimationerror}

%(In this subsection the use of the time parameter $k+1$ aligns with the notations previously introduced, cfr. (\ref{eq:FVIstep1}) and 2(b) in Table \ref{alg:FVI}.)

Recall that the estimate $\big(\Te \We^\ast_{k+1}\big)(\xb)$ of the exact recursion $\big(\T \We^\ast_{k+1}\big)(\xb)$ uses,
for a given state-action pair $(\xb,a)$,
% over the actions $a\in \A$.
the individual Monte-Carlo estimates
\begin{align*}&\frac{1}{M} \sum_{j=1}^{M}\ind_{K}(\xs)+\ind_{A\setminus K}(\xs)\We^\ast_{k+1}(\xs)\end{align*}%\\[-1.4em]
of $\E_{x_k^i}^a\!\left[\ind_K(\xknext)+\ind_{A\setminus K}(\xknext)\We^\ast_{k+1}(\xknext)\right]$.
%\qquad\qquad\right.\\& \left.
%\Big{|}\xknext\sim \Tx{\cdot}{\xb,a}\right]. 
Since the cardinality of the action space $\A$ is finite,
a bound on the error of the estimates for a given state-action pair can lead to a bound on the error over the states $\xb$:
we elaborate on this idea next.

The $M$ random quantities for $1\leq j\leq M$,
$\ind_K(\xs)\allowbreak+  \ind_{A\setminus K}(\xs) \We^\ast_{k+1}(\xs)$, are obtained via independent and identically distributed realizations over the closed interval $[0,1]$.
Hoeffding's inequality \citep{Hoeffding} leads to an upper bound on the deviation of the estimate from the expected value as follows: %\\[-1.4em]
\begin{align*} %\label{prob:reformulated}
&\mathbf{P}\Big{\{}\Big{|}\E_{x_k^i}^a\!\!\left[\ind_K(\xknext)+ \ind_{A\setminus K}(\xknext) \We^\ast_{k+1}(\xknext)%\qquad\qquad\qquad\right.\\\left.  \Big{|} \xknext\sim \Tx{\cdot}{\xb,a} 
 \right] \notag \\&-\frac{1}{M}\sum_{j=1}^{M}\left[ \ind_K(\xs)+ \ind_{A\setminus K}(\xs) \We^\ast_{k+1}(\xs)\right]
\Big{|}\leq \epsilon_1 \Big{\}}%\\
%&\hspace{1cm}
\geq 1- 2e^{-2M(\epsilon_1 )^2},
\end{align*}%\\[-1.4em]%\label{localacc}
where $\epsilon_1$ is the bound on the error.\\*
%Note that empirical recursion at $\xb$ is bounded by $\epsilon_1$, if the Monte-Carlo estimate for all state-action pairs $(\xb,a)$, $\in\A$ is bounded by $\epsilon_1$.
We can then provide a lower bound on the probability that the deviation incurred by the quantity $\big(\Te \We^\ast_{k+1}\big)(\xb)$ is bounded by $\epsilon_1$ via the joint probability of $|\A|$ independent events, as follows:
\begin{align*}
&\hskip-1cm\p{\Big{|}\T \We^\ast_{k+1}(\xb)-\Te\We^\ast_{k+1}(\xb) \Big{|} \leq  \epsilon_1 }\\
&\geq \prod_{a\in\M{A}}
 \mathbf{P}\Big{\{}\Big{|}
 \E_{x_k^i}^a\left[\ind_K(\xknext)+ \ind_{A\setminus K}(\xknext) \We^\ast_{k+1}(\xknext)% \right.\\&\qquad \qquad\qquad \qquad\left.  \Big{|} \xknext\sim \Tx{\cdot}{\xb,a} 
 \right]\notag \\
&\hskip2cm-\frac{1}{M}\!\sum_{j=1}^{M}
  \big[ \ind_K(\xs)\!+\! \ind_{A\setminus K}(\xs) \We^\ast_{k+1}(\xs)\big]
\Big{|}\leq  \epsilon_1 \Big{\}}.
\end{align*}
% bound for all the samples $\xb$ --> union bound
Let us now extend the above probabilistic bound for the error computed at a single point $\xb$ to a bound for the error over all the $N$ base points:
we can express this bound via an empirical $p$-norm defined over the base points,
as follows:%\\[-1.4em]
{\interdisplaylinepenalty=10000
\begin{align} \label{eq:empiricalnorm}
&\|\T \We^\ast_{k+1}\!-\!\Te\We^\ast_{k+1}\|_{p,\hat\eta}%\notag\\&\qquad
=\big(\frac{1}{N}\sum_{i=1}^{N}\left|\T \We^\ast_{k+1}(\xb)-\Te\We^\ast_{k+1}(\xb)\right|^p\big)^{1/p}.
\end{align}}%\\[-1.4em]
This leads to the following result.
\begin{lemma}\label{thm:est1}
For a given error bound $\epsilon_1$ and sample complexity $N$ and $M$,
the estimation error can be probabilistically bounded as follows: \begin{equation}
\p{\|\T \We^\ast_{k+1} -\Te \We^\ast_{k+1}\|_{p,\hat{\eta}}\leq \epsilon_1}\geq 1-\delta_1,
\end{equation} 
where $\delta_1=1-(1-2e^{-2M(\epsilon_1)^2})^{|\M{A}|N}$,
as long as $0<2e^{-2M(\epsilon_1)^2}\leq 1$.
\end{lemma}
%\begin{proof}
%See Appendix \ref{proof:est1}.
%\end{proof}
Notice that for an increasing number of samples $M$,
by Lemma \ref{thm:est1}
the empirical norm of the error as in (\ref{eq:empiricalnorm})
is less than %any given accuracy 
$\epsilon_1$ with a probability that increases to $1$.

\subsubsection{Accuracy of the Empirical Norm}
\label{sec:empnorm}

%(In this subsection the use of the time parameter $k+1$ aligns with the notations previously introduced, cfr. (\ref{eq:FVIstep2}) and 2(c) in Table \ref{alg:FVI}.)

%At each iteration the function $\We^\ast_{k}$ is learned from a finite number of samples.
%
%Assume that each $w\in\M{W}$ maps $\M{X}$ into $[0,1]$.
Let $\T \We^\ast_{k+1}:\M{X}\rightarrow[0,1]$ be an unknown function
and $\eta$ a probability measure on $\M{X}$, $\eta \in M(\M{X})$.
The objective
%of the learning algorithm
is to find a function $w\in \M{W}$ that is close to $\T\We^\ast_{k+1}$ with respect to the following
%weighted $p$-norm, or equivalently with respect to the
\emph{expected loss}: \label{page:loss}%\\[-1.4em]
\begin{align*}%\label{eq:exp_loss}
\inf_{\w\in\mathcal{W}}\!\|w\!-\!\T \We^\ast_{k+1}\|^p_{p,\eta}\! =\inf_{{\w\in\mathcal{W}}} \textstyle\int\limits_{\M{X}}\!\!|w(x)\!-\!\T \We^\ast_{k+1}(x)|^p \eta(x)dx.
 \end{align*}%\\[-1.4em]
This section provides a bound for the error originating from the use of a finite number of samples to evaluate the loss:
this \emph{empirical loss} is defined as%\\[-1.4em]
 \begin{align*}%\label{eq:emp_loss}
 \|\w-\T \We^\ast_{k+1}\|^p_{p,\hat{\eta}}= \frac{1}{N} \sum_{i=1}^{N}|w(\xb)-\T \We^\ast_{k+1}(\xb)|^p ,
\end{align*}%\\[-1.2em]
for a given set of $N$ random variables drawn independently over $A \setminus K$ according to $\xb\sim\eta$.
Let us express a probabilistic bound on this error that holds uniformly over all functions $\w\in\M{\W}$ as follows:
\begin{equation*}\mathbf{P}\big\{\sup_{\mathclap{\w \in \M{W}}}\big{|}\|\w-\T \We^\ast_{k+1}\|^p_{p,\eta}\!\! -\! \|\w\!-\!\T \We^\ast_{k+1}\|^p_{p,\hat{\eta}}\big{|}\geq  \epsilon^p_2  \big\}\!\!\leq \delta_2.\end{equation*}%
%\textcolor[rgb]{1.00,0.50,0.00}%{Notice that
% given a function $w$, \eqref{eq:emp_loss} converges to the expected loss for an increasing number of samples, $N\rightarrow \infty$, by the law of large numbers.
%}
Observe that since the expected and the empirical losses can be reformulated respectively as the mean and the empirical mean of a loss function,
defined informally as $f (x) = |w (x) - \T \We^\ast_{k+1}(x)|^p$, the above problem can be framed
as the uniform convergence of a standard learning problem 
%,Haussler1995,
\citep{Haussler,Pollard1984,Vapnik}.
Resorting to related literature,
results on bounds for the
 %error probability of the estimation of indicator functions in classification theory
%have employed capacity concepts (such as the growth function and the VC dimension) of a function class \citep{Vapnik,Haussler1995}.
%The generalization of the empirical
%risk
%loss
%\footnote{in classification theory one refers to risk instead of loss, as it refers to the "risk" of misclassification.}
%minimization in classification theory to the
error probability of the regression of real-valued functions employ %has led to several uniform convergence results
  capacity concepts (including Rademacher averages, covering numbers, and pseudo dimensions) of a function class \citep{LocalRademacherComplexities,Hoeffding}. 
  %,Pollard1990}.
We focus on pseudo dimensions to deal with the capacity (or the complexity) of the function class.
%The pseudo-dimension of a set of real-valued function can be used to give an upper bound for the covering number of a function space \citep{Mendelson}.
%Moreover the pseudo-dimension
%generalizes the VC dimension and was originally called the VC subgraph \citep{Dudley}.
By results in \citep{Haussler} 
%,Pollard1990} 
and \citep{Pollard1984}, 
we obtain the following uniform convergence bound. 

\begin{lemma}\label{thm:NumberSamples}
Let $\M{W}$ be a set of finitely parameterized, measurable functions on $\M{X}$ taking values in the interval $[0,1]$ with pseudo dimension $\dim_p\left(\M{W} \right)=d<\infty$.
Let $(\xb)_{1\leq i \leq N}$ be generated by $N$ independent draws according to any distribution $\eta$ on $A\setminus K$ and let $p\geq 1$.
Then for any $\epsilon_2>0$ we have that
\begin{align}\label{mythm:1} 
&  \p{\sup_{\w \in \M{W}}\big{|}\|\w-\T \We^\ast_{k+1}\|^p_{p,\eta}-\|\w-\T \We^\ast_{k+1}\|^p_{p,\hat{\eta}}\big{|}\geq  \epsilon^p_2  }%\notag\\&
\leq 4 e(d+1)\left(\frac{32 e}{\epsilon_2^p} \right)^d e^{-\frac{N\epsilon_2^{2p}}{128}}.
\end{align}
\end{lemma}\mbox{ }%\\[-2.5em]
%\begin{proof}
The characterization and computability of the pseudo dimension of a function class is given in Appendix \ref{proof:NumberSamples}.
% is elaborated in Appendix \ref{proof:NumberSamples}.
%\end{proof}
%

\subsubsection{Global Accuracy of Single Iterations}

The error bounds introduced in Lemmas \ref{thm:est1} and \ref{thm:NumberSamples},
together with the inherent Bellman error,
yield an upper bound on the error introduced at each iteration, which is recapitulated in the following statement.
% adapted from \citep{Remi}.
Since the bounds are independent of the distributions $\eta$ and $\Tx{\cdot}{x,a}$,
they are in fact \emph{distribution-free} bounds \citep{LocalRademacherComplexities}.

\begin{theorem}\label{lem:errorsinglestep}
Consider a reach-avoid property defined over a general stochastic system with continuous state space $\M{X}$ and finite action space $\M{A}$. Let $A\subset \M{X}$ and $K\subset \M{X}$ be Borel measurable sets and fix $p\geq1$,
the distribution $\eta \in M(A\setminus K)$ and $\M{W} \subset B(\M{X};1)$.
Pick any $\We^\ast_{k+1}\in B(\M{X};1)$ and let $\Te \We^\ast_{k+1}$ and $\We^\ast_{k}$  be calculated using (\ref{eq:FVIstep1}) and (\ref{eq:FVIstep2}).
For given upper bounds on%\\[-1.4em]
\begin{align}\label{prob:1}
&\mathbf{P}\big\{\|\T \We^\ast_{k+1} -\Te \We^\ast_{k+1} \|_{p,\hat{\eta}}> \epsilon_1\big\}\leq \delta_1, \textmd{ and}\\
&\mathbf{P}\big\{\sup_{\w \in \M{W}}\big{|}\|\w-\T \We^\ast_{k+1}\|_{p,\eta}^p-\|w-\T \We^\ast_{k+1}\|^p_{p,\hat{\eta}}\big{|}> \epsilon_2^p\big\} %\qquad\notag \\& 
\label{prob:2} %\qquad
\leq \delta_2,
\shortintertext{the bound on a single step update is as follows:}
&\p{\|\We^\ast_{k}-\T \We^\ast_{k+1}\|_{p,\eta}> d_{p,\eta}(\T \We^\ast_{k+1},\M{W}) + \epsilon}%\notag\\
 \label{thm:error} %&\qquad
 \leq \delta_1+\delta_2,\end{align}%\\[-1.2em]%
where the optimal approximation is given as a biasing term defined as $d_{p,\eta}(\T \We^\ast_{k+1},\M{W})=\inf_{\w\in \mathcal{\W}} \|\w-\T\We_{k+1}^\ast\|_{p,\eta}$,
and the error $\epsilon$ is given as $\epsilon=2\epsilon_1+2\epsilon_2$.
\end{theorem}
%\begin{proof}
%See Appendix \ref{proof:errorsinglestep}.
%\end{proof}\\*
%
Note that the statement assumes that $\We^\ast_k$  is the unique solution to the optimization problem in (\ref{eq:FVIstep2}).
This strict assumption can be weakened by adding a tolerance term to the theorem.
The quantity $d_{p,\eta}(\T \We^\ast_{k+1},\M{W})$ admits the inherent Bellman error $d_{p,\eta}(\T \M{W},\M{W})$ in \eqref{eq:inherentBellman} as a general upper bound.
\subsection{Error Propagation % over the Whole Horizon
 and Global Error Bounds}% of the FVI Algorithm}
\label{ErrorFVI}
% !TEX root = ALearningApproach.tex

We express a global bound on the accuracy of the FVI algorithm \(\left|\rFVI-r_{x_0}^{\ast}(K,A)\right|\)
by using the probabilistic bounds (derived in Section \ref{ErrorIteration}) on the errors introduced by the approximate one-step mappings
\(|\big(\Te\We_1^\ast\big)(x_0)-\big(\T\We_1^\ast\big)(x_0)|\),
as well as
\(\|\We_1^\ast-\T\We_2^\ast\|_{p,\eta}\),
\(\|\We_2^\ast-\T\We_3^\ast\|_{p,\eta}\),
\ldots,
\(\|\We_{N_t-1}^\ast-\T\We_{N_t}^\ast\|_{p,\eta}\),
and by propagating the error introduced by each single iteration to the successive value iterations,
as done in the next statement.
More precisely, in the following lemma we show that the deviation of the approximate value function $\We^\ast_{k}$ from the optimal value function $\cramped{\T^{N_t-k}\W_{N_t}^\ast}$ can be expressed as a function of this deviation
%of the approximate value function from the optimal value function
at step $k+1$,
plus the approximation error introduced at the $(N_t-k)$th iteration (which has been bounded above).
Recall that the optimal value functions can be written as $\cramped{\W^\ast_{k}=\T^{N_t-k}\W_{N_t}^\ast}$ and that by definition $\cramped{\W_{N_t}^\ast = \We_{N_t}^\ast}$.

%\textcolor{red}{(Formally raise assumption on existence of densities as used next?)}

\begin{lemma}\label{lem:errorpropagation}
Let $\eta$ be the density of a probability distribution with support on $A\setminus K$ and $\txo$ be the density function of the stochastic kernel $\Txo$.
Then
\begin{align}&\|\We^\ast_k-\T^{N_t-k}\We^\ast_{N_t}\|_{p,\eta}\leq \|\We^\ast_{k}-\T \We^\ast_{k+1}\|_{p,\eta}%\qquad\qquad\notag\\&\hspace{2cm}
+B^{\frac{1}{p}}\left\|\We^\ast_{k+1}-\T^{N_t-(k+1)}\We^\ast_{N_t}\right\|_{p,\eta},\notag
\shortintertext{for \(k=0,1,\ldots,N_t-1\), and where $B$ is defined as}
%\begin{equation}
\label{eq:scaling_B}
 &\!\!B=\!\!\! \sup_{\cramped[\scriptstyle]{x_{k+1}\in A \setminus K}} {\int\nolimits_{A\setminus K}}\!\!\! \frac{{\max\limits_{a\in \M{A}}}\ \tx{x_{k+1}\!}{\!x_k,a}\eta(x_k)}{\eta(x_{k+1})}dx_k.%%\\[-2.5em]
 \notag
 \end{align}%
\end{lemma}
%\begin{proof}
%See Appendix \ref{proof:errorpropagation}.
%\end{proof}
%
Putting all the pieces together,
the following theorem provides an expression for the global FVI error bound as the accumulation of the errors from the single iterations over the whole time horizon.
\begin{theorem}\label{thm:MultiStep}
Consider a reach-avoid problem defined on a Markov process with a continuous state space $\M{X}$ and a finite action space $\M{A}$.
The optimal reach-avoid probability $r^{\ast}_{x_{0}}(K,A)$ for a given target set $K$, safe set $A$, initial state $x_0$, and time horizon $N_t$,
is approximated by the quantity $\rFVI$ obtained with the FVI Algorithm in Algorithm \ref{alg:FVI},
which has an accuracy of $\acc$ and a confidence $\Pacc$,
as stated in (\ref{eq:acc}), if the following holds:
\begin{align} 
&\mathbf{P}\bigg{\{}B_0\sum_{k=1}^{N_t-1}B^{\frac{k-1}{p}}\|\We^\ast_{k}-\T \We^\ast_{k+1}\|_{p,\eta}\hspace{3cm} 
\\&\label{eq:Multi_step_state} \hspace{1.5cm}
+|\We^\ast_0(x_0)-\T \We^\ast_1(x_0)|>\acc\bigg{\}}\leq \Pacc,
\shortintertext{where $B$ is given in \eqref{eq:scaling_B}, $B_0$ is defined as}
\label{eq:scaling_state}
&B_0=\sup_{x_1\in A\setminus K} \max_{a\in \M{A}}\frac{\tx{x_1}{x_0,a}}{\eta(x_1)},
\end{align} 
and where $\eta$ is the density of a probability distribution supported on $A\setminus K$
and $\txo$ is the density of the stochastic kernel $\Txo$ of the given Markov process.
\end{theorem}
%\begin{proof}
%See Appendix \ref{proof:MultiStep}.
%\end{proof}
%
%\removed{cor: Initial state distribution}
%\begin{cor}[Initial state distribution]
%Theorem \ref{thm:MultiStep} can be extended to the case of an initial condition taken from a given distribution.
%Consider the density $p_0$ of a given initial state distribution $P_0\in M(\X)$.
%Then,
%%similarly to the proof of Theorem \ref{thm:MultiStep},
%it is possible to show that:%\\[-1.75em]
%\begin{enumerate}
%\item Equation \eqref{eq:Multi_step_state} can be extended to
%%of the accuracy in function of the factors $B$ and $B_0$ equals
%\begin{align*}&\hspace{-1.5em}\mathbf P\Big\{B_0\sum_{k=1}^{N_t-1}B^{\frac{k-1}{p}}\|\We^\ast_{k}-\T \We^\ast_{k+1}\|_{p,\eta} \hspace{2cm}\\&\hspace{-1.5em} 
%+\int_{A\setminus K}\!\!\!\!|\We^\ast_0(x_0)-\T \We^\ast_1(x_0)|p_0(x_0)dx_0>\acc\Big\}%\\
%\leq \Pacc,\end{align*}
%\item and the expression for $B_0$ can be generalized as
%\begin{align*}
%B_0=\sup_{x_1\in A\setminus K}\int_{A\setminus K}\max_{a\in\A}\frac{\tx{x_1}{x_0,a} p_0(x_0)}{\eta(x_1)}dx_0.
%\end{align*}
%\end{enumerate}
%\end{cor}
%
\begin{remark}
Notice that the scaling factor $B$ has a significant influence on the error propagation. 
It is related to the notion of \emph{concentrability} of the future-state distribution \citep{Remi,Future}. %, \textcolor{blue}{which is not a new concept in literature.}
If $B>1$ the error of the algorithm will increase exponentially with the time horizon, 
whereas if $B<1$ the accuracy of the algorithm will depend mostly on the errors in the last few iterations (backwards in time).  
%
%The scaling factor 
It $B$ expresses the maximal concentration of the dynamics (relative to $\eta$) over the relevant set $A\setminus K$  after one transition starting from distribution $\eta$.  
%It is known that for an autonomous, finite-state Markov process, 
%$B\leq 1$ if $\eta$ differs from the stationary distribution of the Markov process over $\X$ only by a constant factor for all $A\setminus K$. 
 \end{remark}

The case study discussed in Section \ref{SectionFive} displays a choice of a density function $\eta$ leading to bounded values for $B$ and $B_0$, respectively. 
The relation between the scaling factor $B$ and the model dynamics for the considered Markov process is also analyzed.  

\subsubsection{Discussion on the Global Error Bounds} 
%\marginpar{\tiny\new{Of all the parts of this paper, I think this is the one that can be reduced the most.}}%
Suppose that we require that the error in the estimation of $r^\ast_{x_0}(K,A)$,
with a confidence at least equal to $\alpha$,
is less than $\acc$, as per \eqref{eq:acc}.
Let us assume that there exist positive values $\epsilon_0, \epsilon_1, \epsilon_2$ such that
the approximation error $\acc$ can be split up into individual bounds based on Theorem \ref{lem:errorsinglestep} and \ref{thm:MultiStep} as follows:%\\[-1.4em]
\begin{align*}%%\label{eq:errorbound} 
\acc=&B_0\textstyle\sum\limits_{k=1}^{N_t-1}B^{\frac{k-1}{p}}d_{p,\eta}\left(\T \M{W},\M{W}\right)%\\[-0.3em] &\hookrightarrow\mbox{\emph{inherent Bellman error (Section \ref{ErrorIteration})}} \\[0em]  &
+2B_0\textstyle\sum\limits_{k=1}^{N_t-1}B^{\frac{k-1}{p}}\left(\epsilon_1\right)%\\[-0.5em]  &\hookrightarrow\mbox{\emph{error on the estimation of $\T \We_{k+1}^{\ast}$ (Section \ref{sec:estimationerror})}}
\\[-0em] 
&+2B_0\textstyle\sum\limits_{k=1}^{N_t-1}B^{\frac{k-1}{p}}\left(\epsilon_2\right)%\\[-0.5em] &\hookrightarrow\mbox{\emph{error on the empirical norm (Section \ref{sec:empnorm})}}\\[-0em]  
+  \epsilon_0,                                        %\hookrightarrow\mbox{\emph{error related to }}|\We^\ast_0(x_0)-\T \We^\ast_1(x_0)| 
\end{align*}%\\[-1.4em]
accounting for, respectively, the inherent Bellman error (Section \ref{ErrorIteration}); the error on the estimation of $\T \We_{k+1}^{\ast}$ (Section \ref{sec:estimationerror}); the error on the empirical norm (Section \ref{sec:empnorm}); and the  error related to $|\We^\ast_0(x_0)-\T \We^\ast_1(x_0)|$.
%\begin{equation}\label{eq:errorbound}\begin{array}{llll}
%\acc=&B_0\sum_{k=1}^{N_t-1}B^{\frac{k-1}{p}}d_{p,\eta}\left(\T \M{W},\M{W}\right)& - &\mbox{inherent Bellman error (Section \ref{ErrorIteration})} \\
%&+2B_0\sum_{k=1}^{N_t-1}B^{\frac{k-1}{p}}\left(\epsilon_1\right)& - &\mbox{error on the estimation of $\T \We_{k+1}^{\ast}$ (Section \ref{sec:estimationerror})}\\
%&+2B_0\sum_{k=1}^{N_t-1}B^{\frac{k-1}{p}}\left(\epsilon_2\right)& - &\mbox{error on the empirical norm (Section \ref{sec:empnorm})}\\
%&+  \epsilon_0                                       & - & \mbox{error related to }|\We^\ast_0(x_0)-\T \We^\ast_1(x_0)| \end{array}
%\end{equation}
We compute a lower bound on the confidence that the approximation error is bounded by $\acc$ by upper bounding the complementary event.
Using a union bounding argument on the probability of invalidating any of the bounding terms in the above equation,
we obtain that the probability that the overall error is larger than $\Delta$ is upper bounded with $\Pacc$, as
\begin{align}\label{eq:confidenceerrorbound}
\Pacc=&0 %\hookrightarrow\mbox{\emph{inherent Bellman error (Section \ref{ErrorIteration})}}  \\[0em] \notag&
\textstyle+(N_t-1)\big(1-(1-2e^{-2M(\epsilon_1)^2})^{|\A|N}\big)\\&%\hookrightarrow\mbox{\emph{Lemma \ref{thm:est1} with sample complexity $M$ and $N$}}\notag\\[0em] \notag&
\textstyle+\cramped{(N_t-1)\Big(4 e (d+1)\left(\frac{32e}{\epsilon_2^p}\right)^de^{-\frac{N\epsilon_{2}^{2p}}{128}}\Big)}\notag \\%\\&\hookrightarrow\mbox{\emph{Lemma \ref{thm:NumberSamples} with }} \pseudo{\M{W}}=d<\infty\notag \\[0em] \notag
&+\textstyle 1-(1-2e^{-2M_0(\epsilon_0)^2})^{|\A|} %\\&\hookrightarrow\mbox{\emph{Lemma \ref{thm:est1} with $M= M_0$ and $N=1$}}%\mbox{Error estimation }|\We^\ast_0(x_0)-\T \We^\ast_1(x_0)|\mbox{,}\\
, \notag  
\end{align}%\\[-1.4em] 
with respectively the inherent Bellman error (Section \ref{ErrorIteration}), the confidence terms from  Lemma \ref{thm:est1} with sample complexity $M$ and $N$, from Lemma
\ref{thm:NumberSamples} with $\pseudo{\M{W}}=d<\infty$,  and from  Lemma \ref{thm:est1} with $M= M_0$ and $N=1$.
%\begin{equation}\label{eq:confidenceerrorbound}\begin{array}{llll}
%\Pacc=&0 & - &\mbox{inherent Bellman error (Section \ref{ErrorIteration})} \\
%&+(N_t-1)\left(1-(1-2e^{-2M(\epsilon_1)^2})^{|\A|N}\right)& - &\mbox{Lemma \ref{thm:est1} with sample complexity $M$ and $N$}\\
%&+(N_t-1)\left(4 e (d+1)\left(\frac{32e}{\epsilon_2^p}\right)^de^{-\frac{N\epsilon_{2}^{2p}}{128}}\right) & - &\mbox{Lemma \ref{thm:NumberSamples} with } \pseudo{\M{W}}=d<\infty\\
%&+ 1-(1-2e^{-2M_0(\epsilon_0)^2})^{|\A|} & - &\mbox{Lemma \ref{thm:est1} with $M= M_0$ and $N=1$}%\mbox{Error estimation }|\We^\ast_0(x_0)-\T \We^\ast_1(x_0)|\mbox{,}\\
%, \end{array}
%\end{equation}
Eqn. (\ref{eq:confidenceerrorbound}) holds as long as $0<2e^{-2M(\epsilon_1)^2}\leq1$ and  $0<2e^{-2M_0(\epsilon_0)^2}\leq1$.
We can observe that it is possible to find finite values for $N$,$M$,$M_0$ such that $1-\alpha> \Pacc$.
Moreover for each choice of positive $\epsilon_0, \epsilon_1, \epsilon_2 > 0$ and confidence $0\leq \alpha< 1$,
the necessary number of samples can be upper bounded by polynomials in $\frac{1}{\epsilon_{0}},\frac{1}{\epsilon_{1}},\frac{1}{\epsilon_{2}}$ and $\frac{1}{1-\alpha}$,
as follows:%\\[-2.5em]
\begin{subequations}
{\small\begin{align}\label{eq:NMM}
N&=\textstyle\Big{\lceil} 128\big(\ln(4e(d+1))+d \ln(32e) \big)\Big(\frac{1}{\epsilon_2}\Big)^{2p}\\ &\textstyle+128dp\Big(\frac{1}{\epsilon_2}\Big)^{2p} \ln\Big(\frac{1}{\epsilon_2}\Big)+128\Big(\frac{1}{\epsilon_2}\Big)^{2p}\ln\Big(\frac{1}{\delta_2}\Big)
\Big{\rceil},\notag\\
M&=\textstyle\Big{\lceil}\frac{1}{2}\Big(\frac{1}{\epsilon_1}\Big)^2\Big(\ln(2|\M{A}|)+\ln(\frac{1}{\delta_1}) +\ln(N)\Big)\Big{\rceil}, \\
M_0&=\textstyle\Big{\lceil}\frac{1}{2}\Big(\frac{1}{\epsilon_0}\Big)^2\Big(\ln\big(2|\M{A}|\big)+\ln\big(\frac{1}{\delta_0}\big)\Big)\Big{\rceil},
\end{align}}
\end{subequations}%\\[-1.4em]
\noindent
with positive parameters $\delta_0, \delta_1, \delta_2 > 0$ such that $1-\alpha=\delta_0+(N_t-1)\delta_1+(N_t-1)\delta_2$ (derivation in Appendix \ref{App:samplecomplexity}).\\[1em]
%
%% Tightness of the error bounds
%- $B_0$, $B$ concentration measures
%- Time horizon
%- Number of samples $N$, $M$
%- Inherent Bellman error
%- The pseudo dimension of the function class
%
%It can be noted that
Note that the above accuracy does not depend on the dimensionality $n$ of the state space.
Therefore the accuracy for models with higher state space dimension will directly depend on the complexity of the function class employed to approximate given value functions.
This is unlike standard grid-based numerical approximation techniques,
such as that proposed in \citep{SA13}, which are known to break down over models with large state-space dimensionality. \\*[1em]
%In grid-based approximation techniques,
%the accuracy obtained %obtained accuracy
% is proportional to the largest diameter of the grid,
%and the number of cells in the grid grows exponentially with the state space dimension $n$.
%Thus, to increase the accuracy by a factor $a$,% the grid diameters has to be decreased proportionally,
%%and 
%the number of cells (and the related memory usage)  shall increase as $a^n$.\\*[1em]
%
Let us add a few comments on the dependency of the accuracy from several design variables of the FVI algorithm.
Firstly, the choice of function class affects both the inherent Bellman error and the pseudo dimension:
while the former gives a measure of how well the the function class $\M{W}$ can represent the value functions $\W^\ast_k$,
the latter is directly related to the complexity of the function class.
The objective is to obtain a low complexity function class that is capable to accurately fit the given value functions.
A good accuracy can be hard to attain when a bad choice of the function class leads to both a large bias (due to the inherent Bellman error) and to a large number of samples.
Secondly, the sample distribution $\eta$ defines, together with the state transitions, the scaling factors $B$ and $B_0$.
In order to minimize the error propagation caused by $B$, the distribution $\eta$ should be ``aligned'' with the model dynamics characterized by the density of the transition kernel.
% [OR] be `quasi' invariant to the state transition}.
Finally, the parameters $N,M,M_0$ follow from the required accuracy demands,
which are reformulated as polynomial functions depending on $\epsilon_0, \epsilon_1, \epsilon_2$ and $\delta_0, \delta_1, \delta_2$, e.g. as in \eqref{eq:NMM}.\\*[1em]
With regards to the single-step errors,
Lemmas \ref{thm:est1} and \ref{thm:NumberSamples} determine a bound uniformly over the whole function class and for any possible probability distribution.
The used distribution-free notions lead to conservative bounds \citep{LocalRademacherComplexities},
which then result in %impractically large bounds for the error of the overall algorithm and to
a large set of required samples.
Notice however that the construction allows to compute the bounds a-priori, before any sample is drawn from the system.
%In \citep{LocalRademacherComplexities} it is proposed to use data-dependent, local error bounds that only hold for distributions of interest. The error bounds are local because instead of looking at the whole function space, a subspace defined by ``a sphere", or a neighborhood within the function class, around the optimized function $\We^\ast_k$ is considered. Another option would be to evaluate the loss using a hold out set (or a testing set) similar to the proposed method in \citep{bbl-msee-00}.
\\*[1em]
In conclusion, the formal probabilistic bounds on the error made by the approximation algorithm show that the algorithm converges in probability to the best approximation for an increasing cardinality of the samples.
\section{Sample-Based Error Bounds %for Dynamic Programming Value functions
}\label{EmpError}
% !TEX root = ALearningApproach.tex
% See page 56 in Thesis
{%\color{blue}
In this section, 
a probabilistic bound on the error of the approximated reach-avoid probability is developed according to a model- and sample-based philosophy. 
This bound can be computed after the reach-avoid probability has been obtained via dynamic programming as time-dependent, approximate value functions $\We^\ast_{k}$ for $0\leq k \leq N_t-1$.  
The obtained bounds are not only sample dependent but also distribution dependent, since knowledge of the transition kernel is necessary to compute scaling factors such as \eqref{eq:scaling_B}. 
Throughout the section it is assumed that the used samples are not correlated with $\We^\ast_{k}$, 
in other words if the estimated value functions $\We^\ast_{k}$ are a result of a sampled-based optimization, 
then the samples used for the bounds in this section are drawn anew and independently.  

The probabilistic bound on the accuracy 
\(\p{\right.\left|\right.\rFVI-r_{x_0}^{\ast}(K,A)\left.\right|>\acc\left.}\leq \Pacc\), 
as given in \eqref{eq:acc}, 
and computed now in a sample-based manner,
includes an empirical estimate of the quantity $\left|\right.\rFVI-r_{x_0}^{\ast}(K,A)\left.\right|$.  
% in $\acc$. 
%
This sample-based estimate is computed as follows: %\\[-2em]  
\begin{compactenum}
\item[a.] collect samples %according to distribution $\eta$, as given in Table \ref{alg:samples},
$\cramped{(x_i)_{1\leq i \leq \Nh}}$ according to (1) in Algorithm\ref{alg:samples} (with $N=\Nh$), 
and subsequently use (2) with $M=\Mh$ to draw both $\cramped{(\xsh_1)_{1\leq j\leq \Mh}}$ and $\cramped{(\xsh_2)_{1\leq j\leq \Mh}}$ 
;   
\item[b.] estimate the single step error \eqref{eq:emp_reg} %$\|\We^\ast_{k} -\Te  \We^\ast_{k+1}\|_{1,\eta}$
 for each $k$; 
\item[c.] estimate the bias \eqref{eq2:bias} %$\|\Te  \We^\ast_{k+1}-\T  \We^\ast_{k+1}\|_{1,\eta}$
 for each $k$;  
 \item[d.] compute the multi-step error as a propagation and a composition of the estimates in (b.) and (c.). 
%, finally leading to $\left|\rFVI-r_{x_0}^{\ast}(K,A)\right|$.
\end{compactenum}
%An additional set of independently drawn samples is used, .
%
%\textcolor{red}{[Can we eliminate this table and refer to table 1, mentioning proper modifications?]}
%
%\begin{table}[htpb]
%\caption{Sample generation for the empirical bound}\label{alg:holdoutsamples}
%\fbox{\begin{minipage}[h]{\linewidth}
%% \vspace{0.05cm}
%% \textbf{Sample generation}\\
% Given safe set $A$, a target set $K$,  and distribution $\eta$,
% generate $\Nh$ base points, and sample with multiplicity of $\Mh$ at each base point as follows:
%  \begin{enumerate}
%  \item
%  Draw $\Nh$ base points $\left( \xh\right)_{1\leq i\leq \Nh}$ from the distribution $\eta$ in $A\setminus K$;
%  \item
%  Draw $\Mh$ samples at each base point $\xh$ and action $a\in \A$ from the stochastic kernel $\Txo$.  
%  Do this two times and denote the set of samples as $\left(\xsh_1\right)_{1\leq j\leq \Mh}$ and $\left(\xsh_2\right)_{1\leq j\leq \Mh}$ respectively. 
%  \end{enumerate}
% \end{minipage}}
% \end{table}
The single step error is estimated as
 \begin{align}\label{eq:emp_reg}
\textstyle\big\|\We^\ast_{k} -\Te  \We^\ast_{k+1}\big\|_{\mathclap{\ \ \ 1,\tilde{\eta}}}
\ \ =\frac{1}{\Nh}
\sum\limits_{i=1}^{\Nh}
\big|\We^\ast_{k}(x^i)-\max\limits_{a\in\A}\Te^{a}_1 \We^\ast_{k+1}(\xh)\big| 
\end{align}%\\[-1.2em]
for all $ 1\leq k\leq N_t-1$. 
The term on the right is the empirical $1$-norm and can be written as a $1$-norm with weighting $\tilde{\eta}$, which is the empirical distribution of $\eta$ resulting from $(x_i)_{1\leq i \leq \Nh}$. 
Let the operator $\Te^{a}_\alpha $ be, 
for $\alpha=1,2$, 
\begin{align}\label{eq:empT1} 
&\textstyle\Te^{a}_{\!\alpha} \We^\ast_{\!\!k+1}(\xh)\!=\!\frac{1}{\tilde{M}}\!\!\sum\limits_{j=1}^{\tilde{M}}\!\ind_K(\xsh_\alpha)\!%\hspace{6em}\\ &\hspace{10em}
+\!\ind_{A\setminus K}(\xsh_\alpha)\We^\ast_{\!k+1}(\xsh_\alpha).%\notag
\end{align} %\\[-1.4em]
The estimate in (\ref{eq:emp_reg}) is biased due to the maximization over the action space, therefore as a second step we estimate a bound on this bias.
The  combination of the two sample sets   $ (\xsh_1 )\cramped{_{1\leq j\leq \Mh}}$ and $ (\xsh_2 )\cramped{_{1\leq j\leq \Mh}}$, 
for each $\xh$ and $a$, 
allows us to estimate this bias for $1\leq k\leq N_t-1$ as 
\begin{align}\label{eq2:bias}
\Big\|\max_{a\in\M{A}} \big|\Te^{a}_1\We^\ast_{k+1}- \Te^a_2\We^\ast_{k+1}\big|\Big\|_{1,\tilde{\eta}}.
\end{align}%\\[-1.2em]
In the following theorem, 
an expression for the bound on \(\p{\right.\left|\right.\rFVI-r_{x_0}^{\ast}(K,A)\left.\right|>\acc\left.}\leq \Pacc\) is derived, 
employing the error propagation technique first used in Section \ref{ErrorFVI},  
the estimates of the single step error above, 
and the bias 
%\eqref{eq:emp_reg}-
\eqref{eq2:bias} in combination with Hoeffding's inequality \citep{Hoeffding}.  
\begin{theorem}\label{thm:DataDependent}
Consider a reach-avoid problem defined on a Markov process with a continuous state space $\M{X}$ and a finite action space $\M{A}$.
The optimal reach-avoid probability $r^{\ast}_{x_{0}}(K,A)$ for a given target set $K$, safe set $A$, initial state $x_0$, and time horizon $N_t$,
is approximated by the quantity $\rFVI$ obtained with the FVI Algorithm in Algorithm \ref{alg:FVI},
which has an accuracy of $\acc$ and a confidence $\Pacc$,
as stated in (\ref{eq:acc}), if the following holds:
\begin{subequations}\begin{align}
&\acc=B_0\textstyle\sum\limits_{k=1}^{N_t-1}B^{k-1}\bigg(
\left\|\We^\ast_{k} -\Te  \We^\ast_{k+1}\right\|_{1,\tilde{\eta}} \!\! \label{eq:emp_acc}%\\& \hspace{2em}
+ \left\|\max_{a\in\M{A}} \left|\Te^{a}_1\We^\ast_{k+1}- \Te^a_2\We^\ast_{k+1}\right|\right\|_{1,\tilde{\eta}}
\bigg)
+B_0\epsilon+\epsilon_0, \\
&\Pacc=e^{-2\frac{\Nh \epsilon^2}{L^2}}-\delta_0, \label{eq:emp_acc_conf} 
\end{align}\end{subequations}
with $L=2\sum_{k=1}^{N_t-1}B^{k-1}$, 
and sample sizes $\Mh$ and $\Nh$ according to the sample sets drawn according to the distribution $\eta$. 
Equation \eqref{eq:emp_acc} includes the estimated error as a combination of (\ref{eq:emp_reg}) and (\ref{eq2:bias}). The scaling factors $B$ and $B_0$ are  computed as in (\ref{eq:scaling_B}) and (\ref{eq:scaling_state}) for the same sampling distribution $\eta$.
The factors $\delta_0$ and $\epsilon_0$ are computed as in Lemma \ref{thm:est1} with $M= M_0$ and $N=1$.   
\end{theorem} 
%\begin{proof}
%See Appendix \ref{proof:emp_acc}.
%\end{proof}
The accuracy $\acc$ depends on two terms, 
%the estimated error $\left|\right.\rFVI-r_{x_0}^{\ast}(K,A)\left.\right|$ which is a 
the propagation of the estimated single-step and bias errors over the time horizon up to $k=1$, 
and the estimation errors $B_0\epsilon+\epsilon_0$ for $k=0$, related to the confidence $\Pacc$. 
 %Where the latter corresponds to the estimation for $k=0$ and the former corresponds to the estimation of $\left|\right.\rFVI-r_{x_0}^{\ast}(K,A)\left.\right|$. 
 % note on extension to initial probability distribution. 

 %\textcolor{blue}{As a corollary to the theorem it can be shown how, for any Markov policy $\mu$ its relative performance $|{r}^\mu_{x_0}(K,A)-r^\ast_{x_0}(K,A)|$ can be bounded probabilistically using a bound on the difference between $r^\mu_{x_0}(K,A)$ and $\rFVI$.}
Suppose that a close-to-optimal policy is given, for example a policy as detailed in Remark \ref{rem:pol} and computed from the series of estimated value functions $\We_k^\ast$. 
Then we know that $\cramped{r_{x_0}^{\ast}}(K,A)\geq {r_{x_0}^{\hat{\mu}^\ast}}(K,A)$, therefore a lower bound on the value of $r_{x_0}^{\hat{\mu}^\ast}(K,A)$ is also a lower bound on $r_{x_0}^{\ast}(K,A)$.   
Note that for a policy $ \mu$, the closed-loop Markov process is time dependent. 
%and autonomous. 
This allows us to estimate $\cramped{r_{x_0}^\mu}(K,A)$ 
%\[
%r_{x_0}^\mu(K,A)=\E_{x_0}^{\mu}\left[\sum_{j\in[0,N_t]}\left(\prod_{i=0}^{j-1}\ind_{A\setminus K}(x_i)\right)\ind_{K}(x_j)\right]
%\]
directly from traces of this autonomous Markov process. 
The deviation of this empirical mean 
%from the mean 
can be bounded probabilistically using Hoeffding's inequality. 
Additionally an upper bound on the deviation $|\rFVI-r_{x_0}^\mu(K,A)|$ can be computed. 
The combination of the bound in Theorem \ref{thm:DataDependent} and of the bound on $|\rFVI-r_{x_0}^\mu(K,A)|$  provide a bound on the performance deviation of $r_{x_0}^\mu(K,A)$:  
the triangle inequality leads to $|r_{x_0}^\ast (K,A)-r_{x_0}^\mu(K,A)|\leq |r_{x_0}^\ast (K,A)-\rFVI|+|\rFVI-r_{x_0}^\mu(K,A)|$. 

In comparison to the a-priori 
bound derived in Section \ref{SectionFour}, 
the sample-based bounds do not depend on the inherent Bellman error and can be shown to be less conservative in general.  
Moreover, they provide insight into the accuracy of the iterations steps. 
However, they give no information about the expected convergence of the algorithm, 
and they can only be computed after a run of the algorithm. 
Similarly to the a-priori  bounds, they do not depend on the dimensionality of the state space and are expected to scale better than those used for grid-based approaches such as \citep{SA13}. 
% The sample based bounds are elucidated further with the case study in the following section.

\section{Case Study and Numerical Experiments}
\label{SectionFive}
%%%%%%%%%%%%%%%%%%%%%%%%%%%%%%%%%%%%%%%%%%%
% Introduction
%%%%%%%%%%%%%%%%%%%%%%%%%%%%%%%%%%%%%%%%%%%
% !TEX root = ALearningApproach.tex
We consider a case study from the literature \citep{Benchmark}, 
where the goal is to maximize the probability that the temperature of two interconnected rooms, while staying within a comfortable range,
reaches a smaller target %an ideal smaller temperature 
range within a given finite time horizon.
The temperature can be affected using local heaters. % systems.
%In this study the outcome is computed by applying the Fitted Value Iteration scheme.
A reach-avoid problem is %intuitively 
set up by selecting as the safe set $A=[17.5\ 22]^2$, as the target set $K=[19.25\ 20.25]^2$,
and a fixed time horizon $N_t=10$.
 The case study was implemented in Matlab R2013b on a notebook with 2.6 GHz Intel Core i5 and 16 GB of RAM.%%%%%%%%%%%%%%%%%%%%%%%%%%%%%%%%%%%%%%%%%%%
% Model
%%%%%%%%%%%%%%%%%%%%%%%%%%%%%%%%%%%%%%%%%%%
\subsection{Model}
The dynamics of the temperature in %of
 the two rooms is described by a Markov model,
with the temperature of the rooms  %elements 
making up the state space $\X=\mathbb{R}^2$,
and where the possible configurations $\{OFF, ON\} = \{0,1\}$ of the two heaters form the finite action space $\A$. 
Hence $\A=\{0,1\}\times\{0,1\}$, and as an example the action related to the first heater in the ON mode and the second in the OFF one is given as $a=[1\ 0]^T\in \A$.
The dynamics at discrete time $k$ 
%, driven by a sequence $n_k$ of i.i.d Gaussian random variables,
is characterized by the following stochastic difference equation:
\begin{align}\label{eq:Benchmark2}
&\xknext=\Ca x_k+\Cb a+\Cc+n_k,
\mbox{ where }\\
&\Ca\!=\!\begin{bsmallmatrix}1-b_1-a_{1,2}&  a_{1,2}\\  a_{2,1}& 1-b_2-a_{2,1} \end{bsmallmatrix},\ \! \Cc  =\begin{bsmallmatrix}b_1 x_a \\ b_2 x_a \end{bsmallmatrix},
\textmd{ and }\Cb=\begin{bsmallmatrix}c_1&0\\0&c_2\end{bsmallmatrix},\notag\end{align}%\\[-1.4em]
and with the following parameters: 
%\\[-1.2em] \begin{itemize}\item 
$x_a$ is the ambient temperature (assumed to be constant),
%\item 
$b_i\geq0$ is a constant for the average heat loss rate of room $i$ to the environment; %,\allowdisplaybreaks%ambient,
%\item 
$a_{ij}\geq0$ is a constant for the average heat exchange rate of room $i$ to room $j \not= i$; %\allowbreak
%\item 
$c_i\geq0$ is a constant for the rate of heat supplied by the heater in room $i$. 
The parameters re instantiated as $b_1=0.0375$, $c_1=0.65$, $x_a=6$, $b_2=0.025$, $c_2=0.6$, and $a_{ij}=0.0625$. 
%\end{itemize}
The noise process $n_k$ is a realization of zero-mean Gaussian % Normal
 random variables with covariance $\nu^2I_{2\times 2}$ (2-dimensional identity matrix $I_{2\times 2}$) and $\nu=0.5$.
%
%\begin{table}[t]
%  \centering
%    \captionof{table}{Parameters of the $2$-room heating case study.
%    %, taken from \cite{Abate2}. 
%    }\label{tab:par}
%   {\rule{\linewidth}{1pt}}
%   \begin{tabular}{lll}
%     % after \\: \hline or \cline{col1-col2} \cline{col3-col4} ...
%  %   $b_i$ & $c_i$  & $x_a$& $a_{ij}$ \\   \hline   \hline
%     $b_1=0.0375$ \hfill& $c_1=0.65$ \hfill& $x_a=6$\\
%     $b_2=0.025$& $c_2=0.6$ & $a_{ij}=0.0625$  
%   \end{tabular}
%   {\rule{\linewidth}{1pt}}
%\end{table}
%The stochastic kernel $\Txo$ is a $2$-dimensional bivariate \textcolor[rgb]{1.00,0.00,0.00}{Gaussian }% normal
% distribution,
Let $\M{N}(\cdot \mid \mu,\Sigma)$ be a $2$-dimensional multivariate normal distribution over $(\X,\M{B}(\X))$ with mean $\mu$ and covariance matrix $\Sigma$,
then the stochastic kernel $\Txo$ is given as%\\[-.4em]
\begin{equation}\label{eq:casestudy_kernel}
\Tx{\cdot}{x,a}=\M{N}\left(\cdot\mid\Ca x+\Cb a+\Cc,\nu^2I_{2\times 2}\right)
\end{equation}%\\[-.4em]
and characterises the probability distribution of the stochastic transitions in \eqref{eq:Benchmark2}.
%Let us rewrite the stochastic kernel \eqref{eq:casestudy_kernel} into a probability density function.
The stochastic kernel \eqref{eq:casestudy_kernel} admits the probability density%\\[-1.4em]
\begin{align} \label{eq:densitytrans}
\tx{y}{x,a}=\frac{1}{\sqrt{|\Sigma|(2\pi)^2}}e^{\left(-\frac{1}{2} \left(y-\bar{\mu}\right)^T \Sigma^{-1} \left(y-\bar{\mu}\right) \right)},
\end{align}%\\[-1.4em]
where $|\cdot|$ denotes the determinant of a matrix, and as before
 the covariance matrix equals $\Sigma=\nu^2I_{2\times 2}$ and the mean value is $\bar \mu=\Ca x+\Cb a+\Cc$.

%%%%%%%%%%%%%%%%%%%%%%%%%%%%%%%%%%%%%%%%%%
% Implement FVI
%%%%%%%%%%%%%%%%%%%%%%%%%%%%%%%%%%%%%%%%%%
\subsection{Application of the Fitted Value Iteration Algorithm}
% intro
    The FVI scheme is implemented as in Algorithm \ref{alg:FVI}, and approximates the solution of the reach-avoid problem. 
    We obtain an approximation of $r^\ast_{x_0}(K,A) = \T^{10} \We^\ast_{10}(x_0)$ by $\Te\We^\ast_1(x_0)$, 
    while using the auxiliary functions $\We^\ast_{9},\allowbreak\ldots, \allowbreak\We^\ast_1$ to approximate $\T \We^\ast_{10},\allowbreak \ldots,\allowbreak \T \We^\ast_2$ in the FVI scheme -- 
    equivalently, function $\We^\ast_k$ approximates $\T^{N_t-k} \We^\ast_{N_t}$. 
    For a given temperature $x_k$ at time instant $k$, the function
    $\We^\ast_k(x_k)$ gives the approximate probability that the consecutive temperature values $x_{k+1},\ldots,x_{N_t}$ will reach the temperature range $[19.25,20.25]^2$ within $N_t-k$ time steps, while staying inside the safe set $[17.5,22]^2$. 

%1. Design of the FVI Algorithm.
In order to apply the FVI algorithm, we select 
a uniform distribution $\eta$ over $A\setminus K$ 
%is selected. 
%a distribution $\eta$ 
to sample from,  
then select a function class $\M{W}$ and a value for $p\geq1$ to solve \eqref{eq:FVIstep2}. 
% 1.1 sample distribution
%    A uniform distribution $\eta$ over $A\setminus K$ is selected. This has a scaling factor $B$ less than $\frac{1}{|\Ca|}|\A|$  (see Appendix \ref{scalingfactor} for the derivation).
% 1.2 function class:
% To approximate the value functions, 
We consider a function class $\M{W}$ composed of Gaussian radial basis function (RBF) neural networks with 50 RBFs with a uniform width of 0.7.  
    The neural network toolbox of Matlab is used to solve the regression problem in \eqref{eq:FVIstep2} as a least-square problem (with $p=2$).
    A neural network with a single layer of hidden units of Gaussian type radial basis functions is proved to be a universal approximator for real-valued functions \citep{hartman1990layered}.
    Furthermore the pseudo dimension of an artificial neural network with $W$ free parameters and $k$ hidden nodes has been upper bounded by $O(W^2k^2)$ \citep{Karpinski1997169,NeuralNetworkLearning}. This means that for any desired precision the required number of samples is bounded by a polynomial in the number of hidden nodes.

The following quantities are obtained for the sample complexities: $N=600$, $M=10^3$, $M_0=10^3$.
    %2. Implementation
The approximate value functions  for  $\We^\ast_{9}, \We^\ast_5 $, and $\We^\ast_1$ are displayed in Fig. \ref{fig:Value}.
On the top plots, a point on the state space is associated with a probability for the reach-avoid property over the given time horizon.
At the bottom, the contour plots (level sets) characterize the set of points that verify the reach-avoid property with a probability at least equal to the given level.
 
\begin{figure}[h!]\centering
% \begin{minipage}{.8\linewidth}\includegraphics[width=\textwidth]{CaseStudyMatlab/Weps.eps}\hfill\\*
 \begin{minipage}{.8\linewidth}\includegraphics[width=\textwidth]{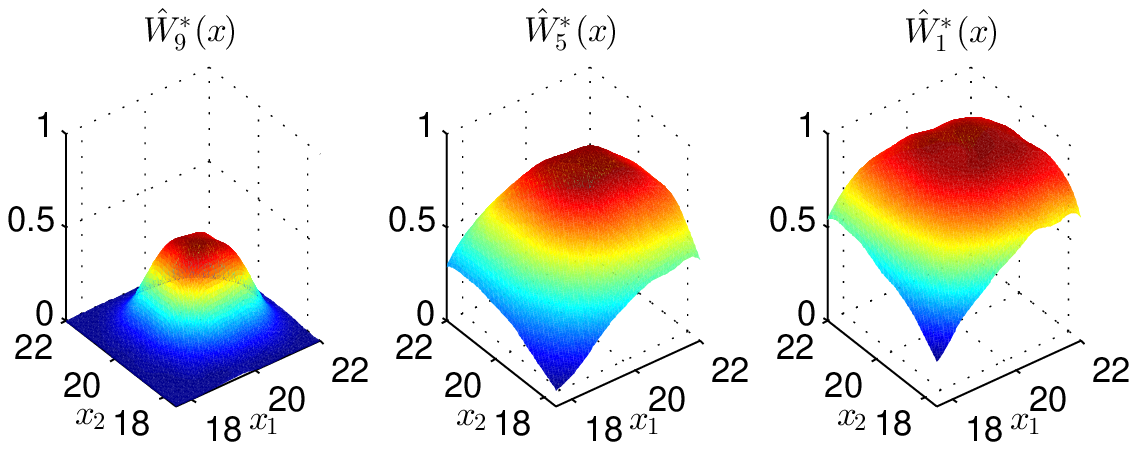}\hfill\\*
% \noindent \includegraphics[width=\textwidth]{CaseStudyMatlab/Wceps.eps} \hfill
 \noindent \includegraphics[width=\textwidth]{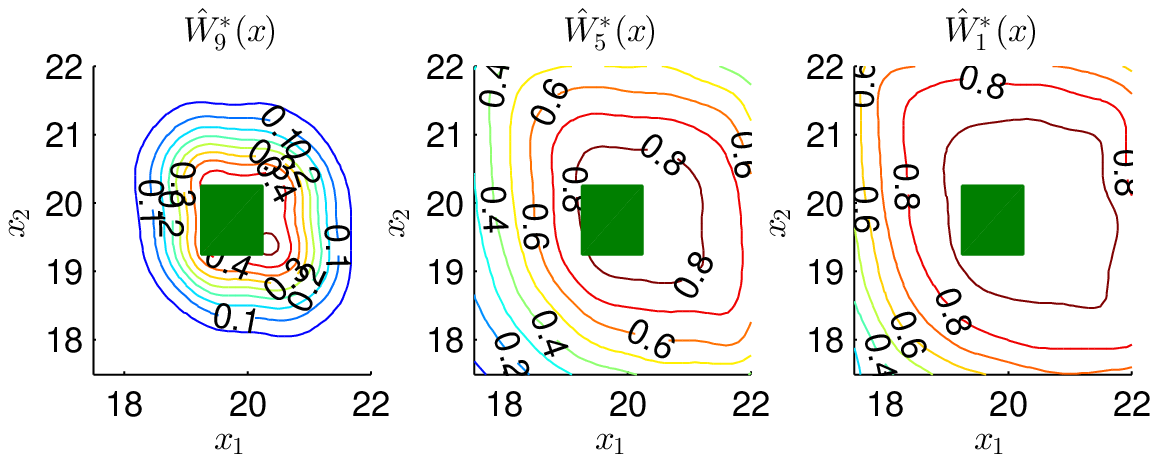} \hfill
\end{minipage}%\\[-2em]

   \caption{
  Function approximations  $\We^\ast_{9}, \We^\ast_5 $, and $\We^\ast_1$ -- level sets (top) and contour plots (bottom).
  The function approximations of the value functions are obtained with the FVI algorithm as in Algorithm \ref{alg:FVI}, using a radial basis function (RBF) neural network with $50$ radial basis functions and a given width of $0.7$,  a uniform sampling distribution $\eta$ over $A\setminus K$ and  $M=10^3$ and $N=600$. The safe set $A$ is $[17.5\ 22]^2$, and reach set $K=[19.25\ 20.25]^2$. The    $\We^\ast_{9}, \We^\ast_5 $, and $\We^\ast_1$ approximate  $\T\W^\ast_{10}, \T^5\W^\ast_{10}$, and $\T^9\W^\ast_{10}$ over the set $A\setminus K$. In the contour plots (bottom) the green squares denote set $K$.
 }\label{fig:Value}
 \end{figure} %In Fig. \ref{fig:pol}
\begin{table}[t]
\caption{Approximate solutions $\rFVI$ are given in the table for several initial conditions $x_0$, together with the related sub-optimal action at the initial time.
}\label{tab:init}
   {\rule{.7\linewidth}{1pt}}
   \centering
\begin{tabular}{p{1.5cm} |c c c c }\\[-1em]
 $x_0$ &  $[19\,\,19]^T$   &  $[20.5\,\,19]^T$ & $[19\,\,20.5]^T$ & $[20.5\,\,20.5]^T$ \\[.1em]%\hline
$\rFVI$&            0.8808  &  0.9454 &   0.9206&    0.9557 \\[.1em]
$a$&  {\small(ON,ON)} &  {\small(OFF,ON)} & 
	  {\small(ON,OFF)}&  {\small(OFF,OFF)}
\\[.1em] \hline \\[-1em]
$x_0$ &  $[18\,\,18]^T$   &  $[21.5 \,\,18]^T$ & $[18 \,\,21.5]^T$ & $[21.5\,\, 21.5]^T$\\[.1em]%\hline
$\rFVI$&     0.5204 &   0.7635 &   0.8596  &  0.8312\\[.1em]
$a$&(ON,ON) &   (ON,OFF) &   (OFF,ON)   & (OFF,OFF)\\[.1em]%\hline
\end{tabular}

   {\rule{.7\linewidth}{1pt}}
\end{table}
Fig. \ref{fig:add1} displays a suboptimal policy $\hat{\mu}^\ast$ that is obtained via the FVI algorithm as discussed in Remark \ref{rem:pol},  
by employing the tree classification method \texttt{ClassificationTree.fit} of \texttt{Matlab}.  
\begin{figure}[h!]\centering
 \includegraphics[width=1\columnwidth]{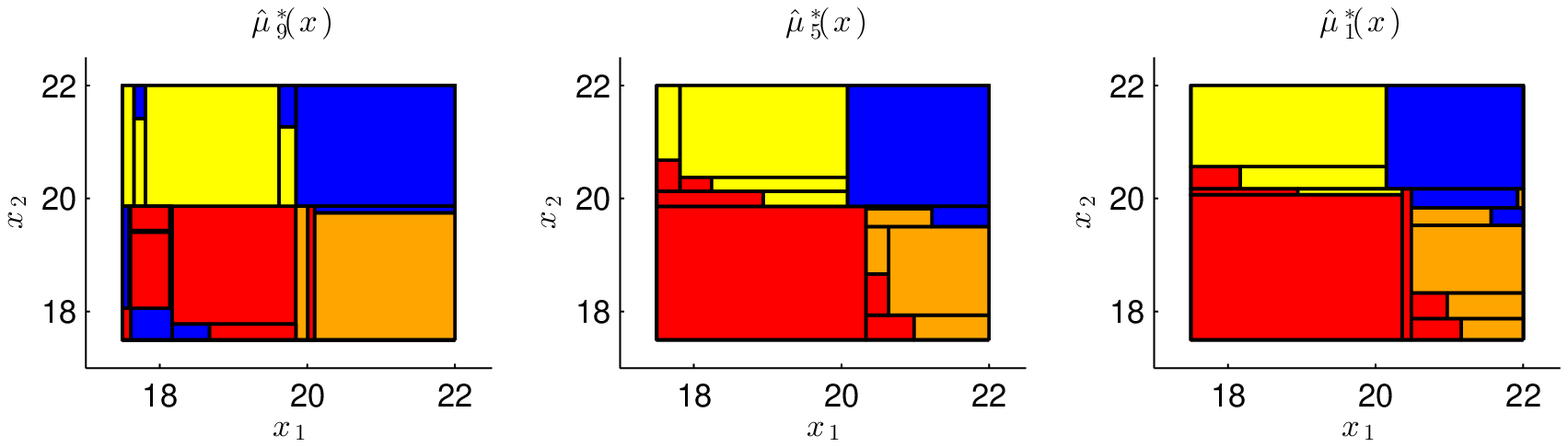}
  \caption{The policy $\hat\mu^\ast$ for $k=9,5,1$ obtained form the same computations as in Fig. \ref{fig:Value}. The action (ON,ON) is labeled in red, (OFF, ON) is orange, (ON,OFF) is yellow, and (OFF,OFF) is blue. }\label{fig:add1}
\end{figure}
%
%The selected policy corresponds to heating actions for the lower temperature regions of the rooms.
Observe that policy $\hat{\mu}^\ast_9$ for $k=9$ is not accurate over the flat regions of $\We^\ast_{9}$ (corresponding to the blue spots in the left side of Fig. \ref{fig:add1} - left plot), 
which are far away from the reach set $K$. Since the average heat loss rate of room 1 is the highest, we expect that the heating should be turned on relatively longer. Fig. \ref{fig:add1} confirms this, i.e. the red (ON,ON) region is not square-shaped as the heaters stay ON for higher temperatures in room 1 than in room 2. 
\subsection{Performance of the Fitted Value Iteration}
We are interested in the performance of the FVI algorithm and in analyzing how the computed accuracy deteriorates over the iterations from $N_t-1$ to $1$. 
Note that the last iteration is of little interest, since it does not include the fitting step. 
The accuracy is computed using the model-based and sample-based bounds of Section \ref{EmpError}. 
Fig. \ref{fig:Sample-base_error} plots the sample-based estimates of the single step error \eqref{eq:emp_reg}, namely $\|\We^\ast_{k} -\Te  \We^\ast_{k+1}\|_{1,\tilde{\eta}}$ 
and of the bias \eqref{eq2:bias}, namely $\big\|\max_{a\in\M{A}} \big|\Te^{a}_1\We^\ast_{k+1}- \Te^a_2\We^\ast_{k+1}\big|\big\|_{1,\tilde{\eta}}$. 
Observe that the values of both   \eqref{eq:emp_reg} and \eqref{eq2:bias} fall in the interval between $3\times 10^{-3}$ and $5\times 10^{-3}$. The bias estimate \eqref{eq2:bias} appears distributed all over this interval, whereas there is a noticeable trend in the plot of \eqref{eq:emp_reg}, which suggests that the first iterations can be fitted more easily. 
% then the later ones.  
%
\begin{figure}[t]
\centering
\resizebox{.8\linewidth}{!}{\begin{tikzpicture}
% results save as 'Results' in CaseStudyMatlab/Results.mat'
\begin{axis}[%
width=.8\linewidth,
height=0.2\linewidth,
scale only axis,axis y line*=left,axis x line*=left,
xmin=0,
xmax=10,
ymin=0.003,title=Estimate,
ymax=0.005,xlabel= $ k $, 
every axis x label/.style={at={(current axis.south)},below=1em},
]
\addplot [color=blue,only marks,mark=x,mark options={solid}]
  table[row sep=crcr] {
9	0.00419633851637198\\
8	0.00453581107111827\\	
7	0.00391773087859635\\	
6	0.0040000600821067\\
5	0.00383088056170553\\	
4	0.00473747586639425\\	
3	0.00464995175278273\\	
2	0.00394770982488436\\	
1	0.0049419311038968\\
		};
\addplot [color=black!50!green,only marks,mark=o,mark options={solid}]
  table[row sep=crcr] {
 9	0.0033665\\	
8	0.00352858299817802\\
7	0.00318633634488082\\	
6	0.00352418706707855\\	
5	0.00395243769944317\\	
4	0.00397050162434925\\	
3	0.00418734555539593\\	
2	0.00433522905401612\\	
1	0.00456060533155798\\
		};
\end{axis}
\end{tikzpicture}}%\vspace{-.5em}
\caption{Sample-based estimation of the single step error  $\|\We^\ast_{k}\allowbreak-\Te  \We^\ast_{k+1}\|_{1,\tilde{\eta}}$ (\textcolor{blue}{ $\times$}) and  $\big\|\max_{a\in\M{A}}\allowbreak \big|\Te^{a}_1\We^\ast_{k+1}\allowbreak- \Te^a_2\We^\ast_{k+1}\big|\big\|_{1,\tilde{\eta}}$ (\textcolor{black!50!green}{$\circ$}) at each iteration for $k=9$ until $k=1$. An independent set of samples of size $\Nh=4\cdot 10^3$ and $\Mh=10^4$ has been used.
 }\label{fig:Sample-base_error}
\end{figure}
In Fig. \ref{fig:TotalErr}, the accuracy of the FVI algorithm propagated over the iterations is given, 
starting from the first iteration $\|\We^\ast_{9}-\T W^\ast_{N_t}\|_{\eta}$ until the last iteration $\|\We^\ast_{1}-\T^{N_t} W^\ast_{N_t}\|_{\eta}$.
This accuracy is computed using Theorem \ref{thm:DataDependent}. The estimates  in Fig. \ref{fig:Sample-base_error} are used to compute the estimate of the accuracy $\|\We^\ast_{k}-\T W^\ast_{N_t}\|_{\eta}$ and the accuracy $\acc$ for a given $\Pacc$.
\begin{figure}[h!]\centering
\resizebox{.8\linewidth}{!}{\begin{tikzpicture}
\begin{axis}[%
width=.8\linewidth,
height=0.2\linewidth,
scale only axis,
xmin=0,title=Accuracy $\acc$,axis y line*=left,axis x line*=left,
xmax=10,
ymode=log,
ymin=0.001,
ymax=1000,
yminorticks=true,
xlabel= $ k $, 
every axis x label/.style={at={(current axis.south)},below=1em},
]
\addplot [color=black!50!green,only marks,mark=o,mark options={solid},forget plot]
  table[row sep=crcr]{
9	0.00747924594411868\\	
8	0.0299790810872909\\
7	0.0990615646383074\\	
6	0.311323814921403\\
5	0.963978693776945\\
4	2.9701952625502\\
3	9.13144621198238\\	
2	28.0560609108734\\
1	86.184288933101\\
};
\addplot [color=blue,only marks,mark=x,mark options={solid},forget plot]
  table[row sep=crcr]{
 9	0.0414099480661942\\
8	0.168129962707759	\\
7	0.55733095649763\\
6	1.75285634344731\\
5	5.42566271309302\\
4	16.708455090703\\
3	51.3632630684878\\	
2	157.807541149486\\
1	484.757646022523\\
};
\end{axis}
\end{tikzpicture}}%\vspace{-.5em}
\caption{The plot shows the accuracy based on the error propagation over the iterations estimated with an hold out set of size $\Nh=4\cdot 10^3$ and $\Mh=10^4$. The accuracy $\acc$, marked as (\textcolor{blue}{\scriptsize $\times$}), is given for $1-\Pacc=0.9$. In the graph, the estimate of the accuracy, $\| \We^\ast_k - \T^{N_t-k}\W^\ast_{N_t}\|$, is given as (\textcolor{black!50!green}{$\circ$}). }\label{fig:TotalErr}\end{figure}
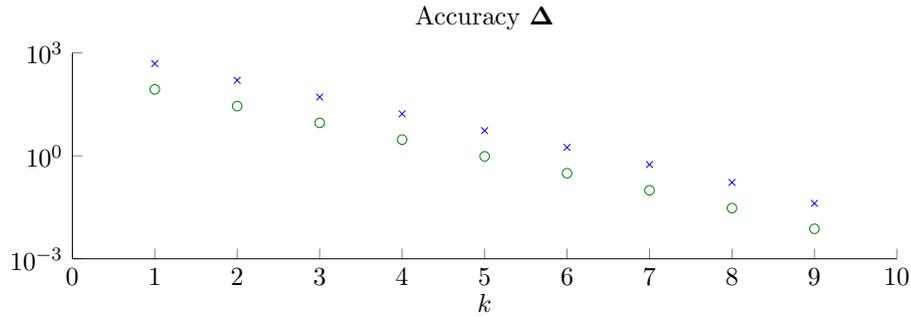
For each iteration step, it can be observed in Fig. \ref{fig:Sample-base_error} that the error caused by estimating the dynamic programming operator $\T$ and by fitting a function is relatively small ($< 10^{-2}$).  
However,  
%when considering the error propagated over the whole horizon, 
the error grows exponentially over the whole horizon: 
%This is caused by the scaling factor $B$, which has been computed numerically and amounts to $3.07$. 
as expected, the accuracy of the algorithm depends strongly on $B$,  
which has been computed numerically and amounts to $3.07$. 
%Therefore good accuracy of the algorithm can only be achieved if the number of iterations $N_t$ is  limited or if the scaling factor is reasonably close to or smaller than $1$. 

%\textcolor{red}{Any error? Is this viable? Any comparison with uniform gridding? Any scalability study? On the error, or at least over the sample complexity? }

\section{Conclusions and Future Work}
\label{Conclusions}
% !TEX root = ALearningApproach.tex
 
This article has investigated the performance of a sample-based approximation scheme for the synthesis of optimal controllers maximizing the probability 
%associated to a property of interest, known as 
of the known ``reach-avoid'' specification. 
%The reach-avoid problem 
%which deals with the maximization of the likelihood that any finite-horizon trajectory of the model enters a given goal set, 
%while avoiding a fixed set of undesired states.
The approximate computational scheme is based on the Fitted Value Iteration algorithm,
which hinges on random sample extractions. \\*
%Formal and explicit probabilistic bounds on the error made by the approximation algorithm show that the algorithm converges in probability to a  neighborhood of the optimal DP solution, bounded by the inherent Bellman error. % best approximation \textcolor{red}{[unclear statement in blue]}. 
%Additionally, the work provides novel sample-based probabilistic error bounds for general dynamic programming solutions of the reach-avoid problem. \\*
We are interested in the non-trivial extension to {\it continuous control} spaces, 
as well as in the assessment of the performance of synthesized approximate policies over the concrete model. 
%For this extended goal, future work will focus on approximating the policy and the value function using two separate function classes. \\*
Finally, the development of better sampling distributions that minimize the error propagation can lead to tighter errors,  
which can be more relevant in practice.  
To this end, the optimal sampling distribution should be used to optimize the scaling factors by resembling more closely the local stochastic kernels.

%\textcolor{red}{%Finally, the development of of empirical bounds (as opposed to the discussed ``distribution free'' ones) promises to lead to tighter errors,
%which can in the end be more meaningful in practice. 
%To this end the formal probabilistic error bounds on the single iteration have to be replaced with sample-based bounds, 
%using empirical estimates of the error \citep{LocalRademacherComplexities} or by using an approach developed for the solution of robust convex optimization problems \citep{calafiore2005uncertain}.  }

%\begin{ack}
% We thank the reviewers of an earlier version of this work for the feedback. 
%We would like to acknowledge the support of the Netherlands Organisation for Scientific Research and of the Dutch Institute of Systems and Control.
%This work has also been supported by the European Commission STREP project MoVeS 257005,
%by the European Commission Marie Curie grant MANTRAS 249295,
%by the European Commission IAPP project AMBI 324432,
%by the European Commission NoE Hycon2 257462,
%and by the NWO VENI grant 016.103.020.\end{ack}
\bibliographystyle{plain}  
\bibliography{BIB}
\clearpage

\appendix
\onecolumn
\section{Proof of Lemma \ref{thm:est1}: Bound on Estimation Error}\label{proof:est1}

We employ results on the concentration of random variables \citep{MAnthony},
which in general raise conditions on a random variable ensuring its realizations to be concentrated around its expectation,
in the sense that the probability of a deviation from the expectation is exponentially small (as a function of the deviation).
Of interest to this work is a known bound holding for sums of bounded and independent random variables \cite{Hoeffding}.

\begin{proposition}[Hoeffding's inequality, \cite{Hoeffding}]\label{thm:Hoeffding}
Suppose that $X_i$, for $i=1,2,\ldots,N,$ are independent random variables
%such that the realizations $\xb$ are
supported on $[0,1]$. Then
$$
\p{\left|\textstyle\sum_{i=1}^{N}\xb-\E{\sum_{i=1}^{N}X_i}\right|\geq N \epsilon}\leq 2 e^{-2N\epsilon^2},
$$
where $\E{\sum_{i=1}^{N}X_i}$ is the mean of the random variable $\sum_{i}X_i$,
whereas the empirical mean is defined as $\sum_{i=1}^{N}\xb$,
where $\xb$ is a realization of $X_i$.
\end{proposition}
Using Proposition \ref{thm:Hoeffding} the proof of Lemma \ref{thm:est1} is provided as follows.\\*
\begin{proof}
Let us express
%the probabilistic error bound as
a probabilistic error bound on the accuracy of the estimate $\Te \We^\ast_{k+1}$ at each base point $\xb$ and given any $a \in \M{A}$ as
\begin{align*}
&\p{\|\T \We^\ast_{k+1} -\Te \We^\ast_{k+1} \|_{p,\hat{\eta}}> \epsilon_1} \leq \delta_1,
\shortintertext{
where we have used the empirical norm based on $\hat{\eta}$. We obtain}
%The above equation is reformulated in a sufficient bound for every state action pair $\xb$, $a$.
% Then a concentration inequality based on the number of samples $M$ (Hoeffding's inequality) is used.
 &\p{\|\T \We^\ast_{k+1}-\Te \We^\ast_{k+1} \|_{p,\hat{\eta}}\leq  \epsilon_1}%\\&
 =
 \p{\|\T \We^\ast_{k+1}-\Te \We^\ast_{k+1} \|_{p,\hat{\eta}}^p\leq \epsilon_1^p} 
 \shortintertext{via definition of the empirical norm in \eqref{eq:empiricalnorm}}
 & =\mathbf{P}\bigg\{ \frac{1}{N}\textstyle\sum\limits_{i=1}^{N}\big{|}\T \We^\ast_{k+1} (\xb)-\Te \We^\ast_{k+1}(\xb)\big{|}^p \leq  \epsilon_1^p\bigg\}
\shortintertext{Note the mutual independence between  the sample sets at different base points $\xb$ given as $\bigcup_{a\in \M{A}} \big(\xs\big)_{1\leq j\leq M} $} 
&\geq\mathbf{P}\bigg\{\textstyle\bigcap\limits_{i=1}^{N}\left\{\Big{|}\T \We^\ast_{k+1}\left(\xb\right)-\Te \We^\ast_{k+1}(\xb)\Big{|}^p\leq  \epsilon_1^p\right\}\bigg\}
%\\&
=\textstyle\prod\limits_{i=1}^{N} \p{\Big{|}\T \We^\ast_{k+1}\left(\xb\right)-\Te \We^\ast_{k+1}(\xb)\Big{|} \leq  \epsilon_1}.
\shortintertext{%
Let us now express the argument of the probability operator as follows}
%Note that the inner term can be rewritten as
%\begin{align*}
&\Big{|}\T \We^\ast_{k+1}\left(\xb\right)-\Te \We^\ast_{k+1}(\xb)\Big{|}\\
&=\Big{|}\max_{a\in\A}\E_{\xknext}\left[\ind_K(\xknext)+\ind_{A\setminus K}(\xknext)\We^\ast_{k+1}(\xknext)\right] \\&-\max_{a\in\A} \frac{1}{M} \textstyle\sum\limits_{j=1}^{M}[\ind_{K}(\xs)+\ind_{A\setminus K}(\xs)\We^\ast_{k+1}(\xs)] \Big{|}\\
\shortintertext{W.r.t $\E_{\xknext}$ defined over random variable $\xknext\sim \Tx{\cdot}{\xb,a}$}
&\leq\max_{a\in\A}\Big{|}\E_{\xknext}\left[\ind_K(\xknext)+\ind_{A\setminus K}(\xknext)\We^\ast_{k+1}(\xknext)\right]\\&- \frac{1}{M} \textstyle\sum\limits_{j=1}^{M} [\ind_{K}(\xs)+\ind_{A\setminus K}(\xs)\We^\ast_{k+1}(\xs)] \Big{|}.
\end{align*}%\\[-1.4em]
Therefore the probability of the last event above can be lower bounded by the probability associated to several independent events over the finite action space, as follows:
\begin{flalign*}
\!\!&\p{\Big{|}\T \We^\ast_{k+1}\left(\xb\right)-\Te \We^\ast_{k+1}(\xb)\Big{|} \leq  \epsilon_1 }\\
&\geq\textstyle\prod\limits_{a\in\M{A}} \mathbf{P}\Big\{ \Big{|}\E_{\xknext}\left[\ind_K(\xknext)+\ind_{A\setminus K}(\xknext)\We^\ast_{k+1}(\xknext)\right] \\&-  \frac{1}{M} \textstyle\sum\limits_{j=1}^{M}\left[\ind_{K}(\xs)+\ind_{A\setminus K}(\xs)\We^\ast_{k+1}(\xs) \right]\Big{|} \leq  \epsilon_1 \Big\}.
\end{flalign*}
For a given base point $\xb\in\M{X}$, action $a\in\M{A}$, and function $ \We^\ast_{k+1}\in\M{\W}$,
define random variables $Z_j$ via their realizations $\ind_K(\xs)+ \ind_{A\setminus K}(\xs) \We^\ast_{k+1}(\xs)$, with $j =1,\ldots,M$.
Since each $\xs$ is independently drawn from $\Tx{\cdot}{\xb,a}$,
the random variables $Z_j$ are independent, identically distributed,
%It can be observed that variables $Z_j$ only
and take values within the closed interval $[0,1]$.
By application of Hoeffding's inequality (as in Proposition \ref{thm:Hoeffding}),
the concentration of the $M$ samples around the expected value of $Z_j$ can be expressed as
\begin{align*}%\label{prob:reformulated}
&\mathbf{P}\Big\{\big{|}\E_{\xknext}\left[\ind_K(\xknext)+\ind_{A\setminus K}(\xknext)\We^\ast_{k+1}(\xknext)\right] \\&-\frac{1}{M}\textstyle\sum\limits_{j=1}^{M}\left[ \ind_K(\xs)+ \ind_{A\setminus K}(\xs) \We^\ast_{k+1}(\xs)\right]
 \big{|}\epsilon_1\Big\}\leq 2e^{-2M(\epsilon_1 )^2}.
\shortintertext{
Therefore as long as $ 0 \leq 2e^{-2M(\epsilon_1 )^2}\leq 1$, it follows that 
$\p{\|\T \We^\ast_{k+1}-\Te \We^\ast_{k+1} \|_{p,\hat{\eta}}\leq  \epsilon_1}\geq \big(1-2e^{-2M(\epsilon_1 )^2}\big)^{N|\A|}$.}%\\[-3em]
\end{align*}%\label{localacc}
\end{proof}
\begin{remark}
As long as we only know that the random variables $Z_j$ are bounded,
the use of Hoeffding's inequality is sufficient.
If we further have information on the variance of $Z_j$,
one can leverage the inequalities of Chebyshev and of Bienaym-Chebyshev \citep{Hoeffding},
or alternatively Bernstein's inequality \citep{Peshkin}:
the former bounds are only function of the variance, whereas the latter inequality depends not only on variance of $Z_j$ but also on its bounded domain.
Upper bounds on either the variance of $Z_j$ or on its range can be derived
exploiting prior knowledge on properties of the function space $\M{W}$ and of the distribution $\Tx{\cdot}{x,a}$.
\end{remark}

%%%%%%%%%%%%%%%%%%%%%%%%%%%

\section{Proof of Lemma \ref{thm:NumberSamples}}\label{proof:NumberSamples}

We derive a general, analytical bound on the error of a single backward recursion using notions from statistical learning theory 
%Haussler1995,
\citep{Haussler,Pollard1984}.
The error bound takes into account that,
for any $\T\We^\ast_{k+1}$,
the optimal fit can be anywhere in the function class.
Furthermore the bound will be distribution-free,
namely holding for any Markov process (with dynamics characterized by $T_x$) and any sample distribution $\eta$ over the set $A\setminus K$.

We exclusively consider function classes $\M{W}\subset B(\X;1)$ endowed with a finite pseudo-dimension:
this includes all finitely-parameterized function classes \citep{Remi}.
The notion of pseudo dimension \citep{Pollard1984,MAnthony,Haussler} expresses the capability of a function class $\M{W}$ to the fit a set of samples.

\begin{proof}
% introduce the simplified notation for $\xb$, and set of samples $\vec{x_i}$.
In order to prove Lemma \ref{thm:NumberSamples}, we show that the inequality in (\ref{mythm:1}) holds for any $\We^\ast_{k+1}\in \M{W}$ at any time instant $k=0,\ldots, N_t-1$.
For the sake of notation in the following we substitute $\We^\ast_{k+1}$ by $W$,
and instead of considering the set of base points $(\xb)_{1\leq i \leq N}$ drawn at the time instant $k$ we simply introduce $\vec{x}=(x^1,\ldots,x^N)$ as a sequence of $N$ independent realizations drawn from a distribution over $A \setminus K$ with density $\eta$.

For any given function $W\in \M{W}$ ,
induce a new function class $l_{\M{W}}=\{|w-\T W|^p:w \in \M{W}\}$ with elements $l_w\in l_{\M{\W}}: l_w=|w-\T\W|^p$.
The inequality in \eqref{mythm:1} can be rewritten over the function class $l_{\M{W}}$ as follows%\\[-1.4em]
\begin{flalign*}
& \mathbf{P}\Big\{\sup_{w \in \M{W}}\big{|}\|w-\T W\|_{p,\eta}^p-\|w-\T W\|_{p,\hat{\eta}}^p\big{|}\geq \epsilon_2^p\Big\}  =\mathbf{P}\Big\{\sup_{l_w \in l_\M{W}}\big{|}\mathbf{E}_{\eta}\left[l_w\right]-\frac{1}{N}\textstyle\sum\limits_{i=1}^Nl_w(x^i) \big{|}\geq \epsilon_2^p \Big\},
\end{flalign*}%\\[-1.4em]
where $\mathbf{E}_{\eta}$ denotes the expected value with respect to $\eta$.
%\textcolor{red}{[better?] I'm unclear about $\mathbf{E}_x\left[l_w(x)\right] \rightarrow \mathbf{E}_\eta \left[l_w\right]$}.
This allows us to use a result in \citep{Pollard1984},
which provides an upper-bound on the probability of the above event as a function of the covering number of the metric space $((l_{\M{W}})_{\vec{x}},\|\cdot\|_{1} )$.
\begin{proposition}[\cite{Pollard1984}]\label{thm:Pollard}
 Let $\M{F}$ be a  permissable set of functions on $\M{X}$ with $0\leq f(x)\leq K $ for all $f\in \M{F}$ and $x \in \M{X}$.
 Let $\vec{x}=(x_1,\ldots,x_N)$ be a sequence of $N$ samples drawn independently from $\M{X}$ according to any distribution on $\M{X}$. Then for all $\epsilon>0$%\\[-1.4em]
\begin{align}
\textstyle \mathbf{P}\Big\{\forall f \in \M{F}:\big|\E{f} - \frac{1}{N} \sum_{\vec{x}} f(x_i) \big|\geq\epsilon\big\}  \leq 4 \E{\M{N}(\epsilon/16,F_{|\vec{x}},\|\cdot\|_1)}e^{-\frac{N\epsilon^2}{128K^2}},
\end{align}%\\[-1.4em]
where the quantity $\M{N}$ will be introduced shortly and where the definition of a permissible set of functions \citep{Pollard1984} includes all finitely parameterized functions.
\end{proposition}
%The empirical distribution $\hat{\eta}$ is defined on the samples $\xb$ for $i=1,\ldots,N$ drawn from the given distribution $\eta$. \textcolor{red}{In $\hat{\eta}$ each set has a measure equal to the fraction of points from $(\xb)_{1\leq i\leq  N}$.}

Let us introduce the concept of covering number of a metric space \citep{Haussler}.
Given a (pseudo-)metric space $(A,\rho)$ and a subset $S$ of $A$, we say that the set $T\subseteq A$  is an $\epsilon$-cover for $S$ (where $\epsilon>0$) if, for every $s\in S$ there is a $t\in T$ such that $\rho(s,t)<\epsilon$.
For a given $\epsilon>0$ we denote the covering number $\M{N}(\epsilon,S,\rho)$ \citep{Haussler} as the cardinality of the smallest $\epsilon$-cover of $S$.
%\footnote{A stricter notion, the proper $\epsilon$-cover, is used by Vapnik \cite{Vapnik}, where $T\subseteq S$. The given definition is conform  \cite{Haussler,Haussler1995}.}
%-> what is a covering number :
%   - pseudo metric space ({l_\M{W}}_{|\vec{x}},\|\cdot\|_{1})
%

For a given set of samples $x^i$ with $i=1,\ldots,N$, the evaluation of a function $l_w\in l_{\M{W}}$ over each of these samples is given as the $N$ dimensional vector in $[0,1]^N$: $(l_w)_{|\vec{x}}=(l_w(x^1),l_w(x^2),\ldots,l_w(x^N))$.
The induced set of vectors is%%\\[-1.4em] 
\begin{align*}
&(l_\M{W})_{|\vec{x}} =\{(l_w)_{|\vec{x}}=(l_w(x^1),l_w(x^2),\ldots,l_w(x^N)), l_w\in l_{\M{W}}\} \subseteq [0,1]^N. \end{align*}%\\[-1.4em]
The minimal $\epsilon$-cover of $((l_\M{W})_{|\vec{x}} ,\|\cdot\|_{1})$ is denoted as $\M{N}(\epsilon,(l_\M{W})_{|\vec{x}},\|\cdot\|_1)$.
%Since $l_{\M{W}}$ is a permissible set of functions on $\X$ with $0\leq l_{w}(x)\leq 1$ for all $l_w\in l_{\M{W}}$ and $x\in\X$, see the permissible set of functions as defined in \cite{Pollard1984}.

The deviation of the expected value from the empirical mean can be bounded using Pollard's proposition \citep{Pollard1984}%\\[-1.4em]
\begin{flalign*}\textstyle
& \textstyle\mathbf{P}\Big\{\sup_{l_w \in l_\M{W}}\big{|}\mathbf{E}_x\left[l_w(x)\right]-\frac{1}{N}\sum_{i=1}^Nl_w(x_i) \big{|}\geq  \epsilon_2^p \Big\}\textstyle\leq  4 \E\left[\M{N}(\epsilon_2^p/16,\left(l_\M{W}\right)_{|\vec{x}},\|\cdot\|_1)\right]e^{-\frac{N( \epsilon_2)^{2p} }{128}}.
\end{flalign*}
The expected value of $\M{N}(\epsilon_2^p/16,\left(l_\M{W}\right)_{|\vec{x}},\|\cdot\|_1)$ is computed over the samples $x^i$ of $\vec{x}$,
drawn independently from a probability distribution with density $\eta$.
Since there is a trivial isometry \citep{Haussler} between  $({l_\M{W}}_{|\vec{x}},\|\cdot\|_{1})$ and $(l_\M{W},\|\cdot\|_{1,\hat{\eta}})$, both spaces have equal covering numbers%\\[-1.4em]
 \begin{align*}
 \M{N}(\epsilon_2^p/16,{l_\M{W}}_{|\vec{x}},\|\cdot\|_{1})=\M{N}(\epsilon_2^p/16,l_\M{W},\|\cdot\|_{1,\hat{\eta}}).\end{align*}%\\[-2.4em]

In practice a value for $\E\left[\M{N}(\epsilon_2^p/16,l_\M{W},\|\cdot\|_{1,\hat{\eta}})\right]$ can be obtained by upper bounding $\M{N}{(\epsilon_2^p/16,l_\M{W},\|\cdot\|_{1,\hat{\eta}})}$ independently of the sample distribution.
For this we introduce the pseudo dimension of a function class,
% Pseudo-dimension
formally defined as follows \cite{Pollard1984,MAnthony,Haussler}.
Suppose $\M{F}$ is a class of functions, $f \in \M{F},f:\M{X}\rightarrow [0,1]$. Then $S \subseteq \M{X}$ is shattered by $\M{F}$ if there are numbers $r_x\in [0,1]$ for $x \in \M{S}$ such that for every $T\subseteq S$ there is some $f_T\in\M{F}$ with the property that $f_{T}\geq r_x$ if $x\in T$ and $f_T<r_x$ if $x \in S\setminus T$. We say that $\M{F}$ has a finite pseudo dimension $\pseudo{\M{F}}=d$ if $d$ is the maximum cardinality of a shattered set. \label{page:pseudodim}
%The covering number  can now be bounded uniformly and independently of the sample distribution and of the sample complexity.

For any distribution $P\in M(\X)$, the packing number \citep{Haussler} and therefore also tho covering number of the metric space $(l_\M{W},\|\cdot\|_{1,P})$ can be upper bounded as a function of the pseudo-dimension and the base of the natural logarithm $e$:
%\textcolor{red}{\begin{proposition}[\cite{Haussler1995}]\label{cor:pseudo} For any set $\M{X}$, any probability distribution $P$ on $\M{X}$, any set $\M{F}$ of $P$-measurable functions on $\M{X}$ taking values in the interval $[0,1]$ with $\pseudo{\M{F}}=d<\infty$  and any $\epsilon >0$
%\begin{flalign*}
%\emph{M}(\epsilon,\M{F},\|\cdot\|_{1,P})\leq e(d+1)\left(\frac{2e}{\epsilon} \right)^d,
%\end{flalign*}
%where $e$ is the base of the natural logarithm.
%\end{proposition}[KEEP OR NOT?]}
%
%The relation between a packing number and a covering number is given in \cite{Haussler};
%\textcolor{red}{\begin{proposition}[Kolmogorov and Tihomirov, {\citep[see][]{Haussler}}]\label{thm:packing}
%If $S$ is a totally bounded subset of the \mbox{(pseudo-)metric} space $(A,\rho)$ then for any $\epsilon>0$,
%$$
%\emph{M}(2\epsilon,S,\rho)\leq \M{N}(\epsilon,S,\rho)\leq \emph{M}(\epsilon,S,\rho).
%$$
%\end{proposition} [KEEP OR NOT?]
%}
%Therefore using Propositions \ref{cor:pseudo} and \ref{thm:packing}, an upper bound on the covering number of the induced class $l_\M{W}$ is expressed as a function of its pseudo dimension.
%As a result we have that
for any
%distribution $P\in M(\M{X})$ and
$\epsilon>0$,%\\[-1.4em]
 \begin{align*}\textstyle
\M{N}(\epsilon,l_\M{W},\|\cdot\|_{1,P})\leq e(d+1)\left(\frac{2e}{\epsilon} \right)^d,  \textmd{ with } \dim_p(l_\M{W})=d.
\end{align*}%\\[-1.4em]
%and $e$ the base of the natural logarithm.
We have proved that a sufficient upper bound is given as%\\[-1.4em]
 \begin{align*} \textstyle
&\textstyle \p{\textstyle\sup_{w \in \M{W}}\big{|}\|w-\T W\|_{p,\eta}^p-\|w-\T W\|_{p,\hat{\eta}}^p\big{|}\geq \epsilon_2^p}\textstyle\leq  4
e(d+1)\left(\frac{32e}{\epsilon_2^p} \right)^d
e^{-\frac{N( \epsilon_2)^{2p} }{128}}.
\end{align*}%\\[-1.4em]
The proof can be concluded by showing that the pseudo dimension $d$ of the induced class $l_\M{W}$ is the same as the pseudo dimension of $\M{W}$.
Let $\left\{ w-\T W : w\in\M{W} \right\}$ be a new function class induced from $\M{W}$.
The invariance properties of the pseudo dimension $\dim_p(\M{W})$ shown in \citep{Haussler} allow to conclude that
$\dim_p(\left\{ w-\T W  \big{|} w\in\M{W} \right\})=\dim_p(\M{W})$.
The induced function class $l_\M{W}$ can then be defined as follows:
$l_\M{W}=\left\{ \left| k\right|^p \big{|} k\in\left\{ w-\T W : w\in\M{W} \right\}\right\}$.
Since it was shown in \citep{Kearns1990} that the pseudo dimension is invariant over function composition ($\left| \cdot\right|^p$),
we conclude that the pseudo dimension is $\dim_p(l_\M{W})=\dim_p(\left\{ w-\T W: w\in\M{W} \right\})=\dim_p(\M{W})=d$.
\end{proof}

%When a function class has a finite pseudo-dimension a lower bound on the sample size for a given confidence $\delta$ and confidence interval $\epsilon$ can be given.
%Let us assume $\M{W}$ to be selected as a finitely parameterized class of functions:
% $$
% \M{W}=\{w_\theta\in B(\M{X})|\theta\in \Theta \},\quad \textmd{dim}(\Theta)\leq\infty
% $$
% This can either be a linear parameterizations $(w_\theta(x)=\theta^T\phi(x))$ or a non-linear parameterizations $(w_\theta(x)=f(x;\theta)$.

\begin{remark}[Computing the pseudo-dimension]
When the function class $\M{W}$ is a vector space of real-valued functions,
the pseudo dimension is equal to the dimensionality of the function class \cite[Theorem 11.4]{NeuralNetworkLearning}.
\citep{NeuralNetworkLearning} elaborates the details of the computation of pseudo dimensions of parameterized function classes,
especially for function classes defined over neural networks.
\end{remark}

Since it is possible to bound the pseudo dimension of $l_{\M{W}}$ (as introduced in the proof)  by the pseudo dimension of $\M{W}$,
this capacity concept has been used to bound the error caused by using an empirical estimate of the weighted $p$-norm.
Notice that for non-parametric function classes, concepts such as covering number or Rademacher average of the function class $l_\M{W}$ can be used instead \citep{LocalRademacherComplexities}.

Let us shortly discuss how the derived bounds can be tightened.
A first option is to circumvent the notion of pseudo dimension and work with the covering numbers in Pollard inequality (Proposition \ref{thm:Pollard}),
however the increase in assumptions on the function class and in overall computations make the gain in accuracy undeserving.
%The main issue is Pollard's inequality.
A second option is to explore alternatives over Pollard inequality in \eqref{thm:Pollard}
with better constants \citep{LocalRademacherComplexities}.
An alternative concentration inequality based on Bernstein's inequality is used in \citep{Peshkin}.
Hoeffding inequality gives a concentration inequality on the sum of bounded random variables, whereas Bernstein inequality gives a tighter bound based on knowledge of both the boundness and the variance of the random variables.
Even with improved constants or alternative inequalities, the error bounds can still result to be conservative for reasonable sample complexities.

%%%%%%%%%%%%%%%%%%%%%%%%%%%

\section{Proof of Theorem \ref{lem:errorsinglestep}}\label{proof:errorsinglestep}

The proof of Theorem \ref{lem:errorsinglestep} is adapted from the proof of the single-step error bound for Fitted Value Iteration with multiple sample batches in \citep{Remi}.
%For completeness, we do give the proof of Theorem \ref{lem:errorsinglestep}, and note that in the reasoning there are some minor differences \cite{Remi}.
\begin{proof}
%To prove Lemma \ref{thm:NumberSamples},
Let us introduce a simplified notation for $\We^\ast_{k+1}$ by replacing it with a general function $W'\in\M{W}$ that minimizes the empirical norm as $W'= \arg\min_{w \in \M{W}}\|w-\Te W\|_{p,\hat{\eta}}$.   
Let us further define a space $\Omega$ for the batch of samples drawn at any of the iterations,
such that at any instant $k$ the realized sample batch
$\omega:=\bigcup_{i\in \{1,\ldots,N\}}\left(\xb \cup \left(\bigcup_{a\in \M{A}} \left(\xs\right)_{1\leq j\leq M}\right)\right)$ is an element of the sample space, $\omega \in \Omega$.

For any given $\epsilon'>0$, %\footnote{$w^\ast$ is different from the $w^\ast$ used in the preliminaries.}
consider a function $w^\ast\in\M{W}$ such that
$\|w^\ast-\T W\|_{p,\eta}\leq \inf_{w\in \M{W}}\|w-\T W\|_{p,\eta}+\epsilon'$
(this in particular holds since $\M{W}$ has been assumed to be close and bounded).
%\textcolor{orange}{Comment: This is a trivial statement. Given a positive $\epsilon'$, if no  function $w^\ast$ exists for which $\|w^\ast-\T W\|_{p,\eta}\leq \inf_{w\in \M{W}}\|w-\T W\|_{p,\eta}+\epsilon'$, then  $\|w-\T W\|_{p,\eta}+\epsilon'$ is a lower bound for   $\|w-\T W\|_{p,\eta},\ w\in \M{W}$. This bound would be strictly larger than the biggest lower bound, which is by definition impossible. Therefore for all $\epsilon'>0$ there will always exist at least one function $w^\ast$ that admits the above inequality. }

The error bound in \eqref{thm:error} holds for a sample realization $\omega$ if the following sequence of inequalities holds simultaneously:
\begin{subequations}\label{ineq:seq}
\begin{align}
\|W'-\T W\|_{p,\eta}&\leq \|W'-\T W\|_{p,\hat{\eta}}+\epsilon_2\label{ineq:seq1}\\
 &\leq \|W'-\Te W\|_{p,\hat{\eta}}+\epsilon_1+\epsilon_2\label{ineq:seq2}\\
 &\leq \| w^\ast-\Te W\|_{p,\hat{\eta}}+\epsilon_1+\epsilon_2\label{ineq:seq3}\\
 &\leq  \|w^\ast-\T W\|_{p,\hat{\eta}} +2\epsilon_1 +\epsilon_2 \label{ineq:seq4}\\
 &\leq  \|w^\ast-\T W\|_{p,\eta}+2\epsilon_1+2\epsilon_2.\label{ineq:seq5}
\end{align}
\end{subequations}
As long as the previous sequence of inequalities is true, the following one also holds:
 $$
 \|W'-\T W\|_{p,\eta}\leq  \inf_{w \in \M{W}} \|w-\T W\|_{p,\eta}+2\epsilon_1+2\epsilon_2+\epsilon'.
$$
We claim that the sequence of inequalities holds with a probability at least $1-(\delta_1+\delta_2)$.
Since there exists a function $w^\ast$ for any $\epsilon'>0$ it follows with a probability at least $1-(\delta_1+\delta_2)$ that
 $$
 \|W'-\T W\|_{p,\eta}\leq d_{p,\eta}(\T W,\M{W})+2\epsilon_1+2\epsilon_2.
$$
By the union bound argument \citep{MAnthony},
the probability of the union of events can be bounded by the sum of the probabilities of the single events.
Using this argument it is possible to define a lower bound on the probability associated with the simultaneous occurrence of the five inequalities in (\ref{ineq:seq}).
We first show that the third inequality is always true.
Then we give the probability associated to the first inequality (\ref{ineq:seq1}) and the fifth (\ref{ineq:seq5}) (this is based on \eqref{prob:2}).
Afterwards we provide an upper bound on the probability associated to the second and fourth inequalities (\ref{ineq:seq2}),(\ref{ineq:seq4}),
based on the bound given in \eqref{prob:1}.

%% E3 --> ALWAYS TRUE
The third inequality \eqref{ineq:seq3} is true for the whole sample space $\Omega$ due to the choice of $W'$. For all functions $w$ in $\M{W}$ it follows that
$\|W'-\Te W\|_{p,\hat \eta}\leq\|w-\Te W\|_{p,\hat \eta}$ holds, because $W'= \arg\min_{w \in \M{W}}\|w-\Te W\|_{p,\hat{\eta}}$.

%% E1 and E5
The first and last inequalities (\ref{ineq:seq1}),(\ref{ineq:seq5}) bound the deviation between the empirical loss and the expected loss. This can be bounded with the worst case error.
Firstly we observe that the inequality
$$\Big{|}\|w -\T W\|_{p,\hat{\eta}}- \|w -\T W\|_{p,\eta}\Big{|}^p\leq \Big{|}\|w-\T W\|_{p,\eta}^p-\|w-\T W\|_{p,\hat{\eta}}^p\Big{|}$$
is always true.
In the case that $\|w -\T W\|_{p,\hat{\eta}}\leq\|w -\T W\|_{p,\eta}$ then
\begin{align*}
\Big{|}\|w -\T W\|_{p,\hat{\eta}}- \|w -\T W\|_{p,\eta}\Big{|}& = \|w -\T W\|_{p,\hat{\eta}}- \|w -\T W\|_{p,\eta}\\
 \|w -\T W\|_{p,\hat{\eta}}& =\left( \|w -\T W\|_{p,\hat{\eta}}- \|w -\T W\|_{p,\eta}\right) + \|w -\T W\|_{p,\eta}\\
\|w -\T W\|^p_{p,\hat{\eta}}& =\left(\left( \|w -\T W\|_{p,\hat{\eta}}- \|w -\T W\|_{p,\eta}\right) + \|w -\T W\|_{p,\eta}\right)^p\\
\|w -\T W\|^p_{p,\hat{\eta}}&\geq \left( \|w -\T W\|_{p,\hat{\eta}}- \|w -\T W\|_{p,\eta}\right)^p + \|w -\T W\|_{p,\eta}^p\\
\|w -\T W\|^p_{p,\hat{\eta}} -\|w -\T W\|_{p,\eta}^p &\geq \left( \|w -\T W\|_{p,\hat{\eta}}- \|w -\T W\|_{p,\eta}\right)^p\\
\Big{|}\|w -\T W\|^p_{p,\hat{\eta}} -\|w -\T W\|_{p,\eta}^p \Big{|}& \geq \Big{|} \|w -\T W\|_{p,\hat{\eta}}- \|w -\T W\|_{p,\eta}\Big{|}^p
\end{align*}
On the other hand,
for the case when $\|w -\T W\|_{p,\hat{\eta}}>\|w -\T W\|_{p,\eta}$ a similar argument can be used.
We can then observe that
 $$\Big{|}\|w^\ast-\T W\|_{p,\hat{\eta}}- \|w^\ast-\T W\|_{p,\eta}\Big{|}^p\leq\sup_{w \in \M{W}}\Big{|}\|w-\T W\|_{p,\eta}^p-\|w-\T W\|^p_{p,\hat{\eta}}\Big{|},$$
 and that
 $$\Big{|}\|W'-\T W\|_{p,\eta}-\|W'-\T W\|_{p,\hat{\eta}}\Big{|}^p\leq\sup_{w \in \M{W}}\Big{|}\|w-\T W\|_{p,\eta}^p-\|w-\T W\|_{p,\hat{\eta}}^p\Big{|}.$$ 
Given two functions $w^\ast$ and $W'$ define events $A_1$ and $A_2$
\begin{equation*}
\begin{aligned}
A_1:\  \epsilon_2^p< \Big{|}\|w^\ast-\T W\|_{p,\hat{\eta}}- \|w^\ast-\T W\|_{p,\eta}\Big{|}^p,\quad
A_2:\ \epsilon_2^p<\Big{|}\|W'-\T W\|_{p,\eta}-\|W'-\T W\|_{p,\hat{\eta}}\Big{|}^p. \end{aligned}
\end{equation*}
 Observe that the event sets $A_1$ and $A_2$ are subsets of the more general event $B$ defined as
 \begin{equation*}
\begin{aligned}B:&\quad\epsilon_2^p<\sup_{w \in \M{W}}\Big{|}\|w-\T W\|_{p,\eta}^p-\|w-\T W\|_{p,\hat{\eta}}^p\Big{|}. \end{aligned}
\end{equation*}  
 Thus it follows that for any $\epsilon_2>0$:  $\p{A_1\cup A_2}\leq \p{B}$ and, 
based on \eqref{prob:2}, we have
\begin{align*}\textstyle
 \p{\Big\{\big{|}\|W'-\T W\|_{p,\eta}-\|W'-\T W\|_{p,\hat{\eta}}\big{|}>\epsilon_2\Big\}\cup \Big\{ \big{|}\|w^\ast-\T W\|_{p,\hat{\eta}}- \|w^\ast-\T W\|_{p,\eta}\big{|}>\epsilon_2\Big\}}\notag \allowdisplaybreaks[0] \\
\leq \p{\sup_{w\in\M{W}}\big{|}\|w-\T W\|_{p,\eta}^p-\|w-\T W\|_{p,\hat{\eta}}^p\big{|}> \epsilon_2^p}\leq \delta_2.
\end{align*}
Thus the probability that the inequalities (\ref{ineq:seq1}) and (\ref{ineq:seq5}) do not hold is less then $\delta_2$.

%%% E2 and E4 %%%
The second and fourth inequalities (\ref{ineq:seq2}),(\ref{ineq:seq4}) depend the accuracy of the estimation of the backward recursion at each base point $\xb$.
Employing the inequality $\left| \|w-g\|_{p,\hat{\eta}}-\|w-h\|_{p,\hat{\eta}}\right|\leq \|g-h\|_{p,\hat{\eta}}$,
we can see that
\begin{align*}
\left|\|W'-\T W\|_{p,\hat{\eta}}- \|W'-\Te W\|_{p,\hat{\eta}}\right|\leq \|\T W-\Te W\|_{p,\hat{\eta}}, 
\shortintertext{ and } 
\left| \| w^\ast-\Te W\|_{p,\hat{\eta}}- \|w^\ast-\T W\|_{p,\hat{\eta}}\right|\leq  \|\T W-\Te W\|_{p,\hat{\eta}}.
\end{align*}
For every sample set $\omega$ the inequalities (\ref{ineq:seq2}),(\ref{ineq:seq4}) apply if $\|\T W-\Te W\|_{p,\hat{\eta}}\leq \epsilon_1$.
Thus%\\[-2em]
 \begin{align}\textstyle
\p{\textstyle\left\{\|W'-\T W\|_{p,\hat{\eta}}- \|W'-\Te W\|_{p,\hat{\eta}}> \epsilon_1\right\}\cup \left\{\| w^\ast-\Te W\|_{p,\hat{\eta}}- \|w^\ast-\T W\|_{p,\hat{\eta}}>\epsilon_1\right\}}\notag\\\leq \p{\|\T W -\Te W \|_{p,\hat{\eta}}> \epsilon_1}\leq \delta_1.
\end{align} The probability that at least one of the inequalities in \eqref{ineq:seq} does not hold can be expressed using the union bound as $\delta_1+\delta_2$. Thus the sequence of inequalities holds with at least a probability of $1-\delta_1-\delta_2$.
\end{proof}

%%%%%%%%%%%%%%%%%%%%%%%%%%%

%\newpage

\section{Proof of Lemma \ref{lem:errorpropagation}} \label{proof:errorpropagation}
%We prove the {\bf Lemma} \ref{lem:errorpropagation} from
%Section~4.1:

\noindent \begin{proof}
%Introduce the function $\We^\ast_{k+1}$,
Let us set up the following chain of inequalities:
\begin{align*}
&\|\T^{N_t-k}\We^\ast_{N_t}-\We^\ast_k\|_{p,\eta}\allowdisplaybreaks\\
&=\mbox{[ Add and subtract function $\T\We^\ast_{k+1}$ ]}\\
&=\|\T \left(\T^{N_t-k-1}\We^\ast_{N_t}-\We^\ast_{k+1}+\We^\ast_{k+1}\right)-\We^\ast_k\|_{p,\eta}\\
&=\mbox{[ Definition of $\T$ in (\ref{eq:T}), where we have considered a single $x_k \sim \eta$ ]}\\
&=\Big\|\max_{a\in\M{A}}\mathbf{E}\Big{[} \ind_{K}(x_{k+1})+\ind_{A\setminus K}(x_{k+1}) \left(\T^{N_t-k-1}\We^\ast_{N_t}-\We^\ast_{k+1}+\We^\ast_{k+1}\right)(x_{k+1})
\big{|}x_{k+1}\sim \Tx{\cdot}{x_k,a}\Big{]} \allowdisplaybreaks\\
&\qquad -\We^\ast_k\Big\|_{p,\eta}\\
&= \mbox{ [ $\max\E{[\xi_1+\xi_2]}\leq \max\E{|\xi_1|}+\max\E{|\xi_2|}$ ] }\allowdisplaybreaks\\
&\leq\left\|\max_{a\in\M{A}}\mathbf{E}\Big{|}\ind_{A\setminus K}(x_{k+1}) \left(\T^{N_t-k-1}\We^\ast_{N_t}-\We^\ast_{k+1}\right)(x_{k+1})|x_{k+1}\sim \Tx{\cdot}{x_k,a}\Big{|}\right.\\
&\left.\qquad +\max_{a\in\M{A}}\mathbf{E}\Big{|} \ind_{K}(x_{k+1})+\ind_{A\setminus K}(x_{k+1}) \We^\ast_{k+1}(x_{k+1})|x_{k+1}\sim \Tx{\cdot}{x_k,a}\Big{|}-\We^\ast_k\right\|_{p,\eta}\allowdisplaybreaks\\
&=\mbox{[ Triangular inequality]}\allowdisplaybreaks\\
&\leq\left\|\max_{a\in\M{A}}\mathbf{E}\Big{|} \ind_{A\setminus K}(x_{k+1}) \left(\T^{N_t-k-1}\We^\ast_{N_t}-\We^\ast_{k+1}\right)(x_{k+1})|x_{k+1}\sim \Tx{\cdot}{x_k,a}\Big{|} \right\|_{p,\eta}\\
&\qquad+\left\|\max_{a\in\M{A}}\mathbf{E}\Big{|}  \ind_{K}(x_{k+1})+\ind_{A\setminus K}(x_{k+1}) \We^\ast_{k+1}(x_{k+1})|x_{k+1}\sim \Tx{\cdot}{x_k,a}\Big{|} -\We^\ast_k\right\|_{p,\eta}\\
&=\mbox{[ Definition of $\T$ in (\ref{eq:T}) ]}\\
&=\left\|\max_{a\in\M{A}}\mathbf{E}\Big{[}\left|\ind_{A\setminus K}(x_{k+1}) \left( \T^{N_t-k-1}\We^\ast_{N_t}-\We^\ast_{k+1}\right) (x_{k+1}) \right|\,|x_{k+1}\sim \Tx{\cdot}{x_k,a}\Big{]}\right\|_{p,\eta}\allowdisplaybreaks\\
&\qquad+\left\|\T \We^\ast_{k+1}-\We^\ast_k\right\|_{p,\eta}\\
&=\mbox{[ Introduce density function $\tx{x_{k+1}}{x_k,a}$ for kernel $T_x$ ]}
\\
&=\left\|\max_{a\in\M{A}}\int_{\X}\left|\ind_{A\setminus K}(x_{k+1})\left(\T^{N_t-k-1}\We^\ast_{N_t}-\We^\ast_{k+1}\right) (x_{k+1}) \right| \txo(x_{k+1}|x_k,a)dx_{k+1}\right\|_{p,\eta}\left\|\T \We^\ast_{k+1}-\We^\ast_k\right\|_{p,\eta}.
\shortintertext{%
Let us now show that the first term is bounded by \(B^{\frac{1}{p}} \left\|\T^{N_t-k-1}\We^\ast_{N_t}-\We^\ast_{k+1}\right\|_{p,\eta}\):
}
%\textcolor{red{(for simplicity we drop the dependence on $x_{k+1}$)}:
%
&\left\|\max_{a\in\M{A}}\int_{\X}\left|\ind_{A\setminus K}(x_{k+1})\left(
\T^{N_t-k-1}\We^\ast_{N_t}-\We^\ast_{k+1}\right)\right|\txo(x_{k+1}|x_k,a)dx_{k+1}\right\|_{p,\eta}\\
&=\mbox{[ monotonicity of $L_p$-norms with respect to a probability measure ]}\\
&\leq\left\|\max_{a\in\M{A}}\left(\int_{\X}\left|\ind_{A\setminus K}(x_{k+1})\left(
\T^{N_t-k-1}\We^\ast_{N_t}-\We^\ast_{k+1}\right)\right|^p\txo(x_{k+1}|x_k,a)dx_{k+1}\right)^{\frac{1}{p}}\right\|_{p,\eta}\allowdisplaybreaks\\
& = \mbox{[ Express the $\eta$-weighted $p$-norm over $A\setminus K$ ]}\allowdisplaybreaks \\
&\textstyle=\left(\int_{A\setminus K} \left|\max_{a\in\M{A}}   \left|\int_{A\setminus K}   \left|
\T^{N_t-k-1}\We^\ast_{N_t}-\We^\ast_{k+1}\right|^p     \txo(x_{k+1}|x_k,a)dx_{k+1}\right|^{\frac{1}{p}}    \right|^{p}\eta(x_k) dx_k\right)^{\frac{1}{p}}\allowdisplaybreaks\\
&= \textstyle\left(\int_{A\setminus K}\max_{a\in\M{A}}\left|\left|\int_{A\setminus K}\left|
\T^{N_t-k-1}\We^\ast_{N_t}-\We^\ast_{k+1}\right|^p\txo(x_{k+1}|x_k,a)dx_{k+1}\right|^{\frac{1}{p}}\right|^{p}\eta(x_k)  dx_k\right)^{\frac{1}{p}}\allowdisplaybreaks\\
&=\textstyle \left(\int_{A\setminus K}\max_{a\in\M{A}} \int_{A\setminus K}\left|
\T^{N_t-k-1}\We^\ast_{N_t}-\We^\ast_{k+1}\right|^p\txo(x_{k+1}|x_k,a) dx_{k+1} \eta(x_k)dx_k\right)^{\frac{1}{p}}\allowdisplaybreaks\\
&\textstyle\leq
\left(\int_{A\setminus K} \int_{A\setminus K}\max_{a\in\M{A}}\left(\left|
\T^{N_t-k-1}\We^\ast_{N_t}-\We^\ast_{k+1}\right|^p\txo(x_{k+1}|x_k,a)\right)dx_{k+1}\eta(x_k) dx_k\right)^{\frac{1}{p}}\allowdisplaybreaks\\
&=\mbox{[ Introduce dummy term $\frac{\eta(x_{k+1})}{\eta(x_{k+1})}$, which is defined over $x_{k+1}\in A\setminus K$ ]}\allowdisplaybreaks\\
&=\textstyle
\left(\int_{A\setminus K} \int_{A\setminus K}\max_{a\in\M{A}}\left(\left|
\T^{N_t-k-1}\We^\ast_{N_t}-\We^\ast_{k+1}\right|^p  \txo(x_{k+1}|x_k,a) \right) \frac{\eta(x_k)}{\eta(x_{k+1})}dx_k\eta(x_{k+1}) dx_{k+1}\right)^{\frac{1}{p}}\allowdisplaybreaks\\
&=\mbox{[ Recall that
$\left|\T^{N_t-k-1}\We^\ast_{N_t}-\We^\ast_{k+1}\right|^p$ is only a function of $x_{k+1}$ ]}\allowdisplaybreaks\\
&=\textstyle
\left(\int_{A\setminus K} \left|
\T^{N_t-k-1}\We^\ast_{N_t}-\We^\ast_{k+1}\right|^p\int_{A\setminus K} \left(\max_{a\in\M{A}} \frac{\txo(x_{k+1}|x_k,a)\eta(x_k)}{\eta(x_{k+1})}\right) dx_k\eta(x_{k+1}) dx_{k+1}\right)^{\frac{1}{p}}.
\end{align*}
Introduce now the upper bound on $\int_{A\setminus K} \left(\max_{a\in\M{A}} \frac{\txo(x_{k+1}|x_k,a)\eta(x_k)}{\eta(x_{k+1})} \right) dx_k$  over the domain $A\setminus K$  as
 \( B=\sup_{x_{k+1}\in A \setminus K}\int_{A\setminus K} \max_{a\in \M{A}}\frac{\txo(x_{k+1}|x_k,a)\eta(x_k)}{\eta(x_{k+1})} dx_k\),~obtaining
\begin{align*}
&\textstyle\leq
\left(\int_{A\setminus K} \left|
\T^{N_t-k-1}\We^\ast_{N_t}-\We^\ast_{k+1}\right|^p B\eta(x_{k+1})dx_{k+1}\right)^{\frac{1}{p}}=B^{\frac{1}{p}}\|
\T^{N_t-k-1}\We^\ast_{N_t}-\We^\ast_{k+1}\|_{p,\eta}. \end{align*}
We have finally shown that
\(
\left\|\T^{N_t-k}\We^\ast_{N_t}-\We^\ast_k\right\|_{p,\eta}\leq \left\|\T \We^\ast_{k+1}-\We^\ast_k\right\|_{p,\eta}+
B^{\frac{1}{p}}\left\| \T^{N_t-k-1}\We^\ast_{N_t}-\We^\ast_{k+1}\right\|_{p,\eta}.\)%\\[-3em]
\mbox{ \mbox{ }}
\end{proof}

%%%%%%%%%%%%%%%%%%%%%%%%%%%

\section{Proof of Theorem \ref{thm:MultiStep}}\label{proof:MultiStep}
%The {\bf Theorem} \ref{thm:MultiStep} is proved as follows.

\begin{proof}
If we estimate the quantity $r^{\ast}_{x_0}(K,A) = \left(\T^{N_t}\W_{N_t}^\ast\right)(x_0) = \left(\T^{N_t}\We^\ast_{N_t}\right)(x_0)$ by $\left(\Te \We^\ast_{1}\right)(x_0)$,
then we have that $\rFVI=\ind_K(x_0)+\ind_{A\setminus K}(x_0)\left(\Te \We^\ast_{1}\right)(x_0)$.
The absolute deviation of the approximated $\rFVI $ from the exact $r_{x_0}^\ast(K,A)$ is given as
\[
\left|\rFVI-r_{x_0}^\ast(K,A)\right|=\left|\left(\T^{N_t}\We^\ast_{N_t}\right)(x_0)- \left(\Te \We^\ast_{1}\right)(x_0)\right|.
\]

The objective is to present this error as a function of the errors introduced by the approximate mappings  \(\left|\left(\T^{N_t}\We^\ast_{N_t}\right)(x_0)- \left(\Te \We^\ast_{1}\right)(x_0)\right|\),
as well as of the quantities
\(\|\We_1^\ast-\T\We_2^\ast\|_{p,\eta}\),
\(\|\We_2^\ast-\T\We_3^\ast\|_{p,\eta}\),
\ldots,
\(\|\We_{N_t-1}^\ast-\T\We_{N_t}^\ast\|_{p,\eta}\).

To this end, we first express a bound on
\(\left|\left(\T^{N_t}\We^\ast_{N_t}\right)(x_0)- \left(\Te \We^\ast_{1}\right)(x_0)\right|\) as a function of \(|\left(\Te \left(\T\We_1^\ast\right)(x_0) - \We^\ast_{1}\right)(x_0)|\)
and of \(\left\|\T^{N_t-1}\We_{N_t}^\ast-\We_1^\ast\right\|_{p,\eta}\).
Then Lemma \ref{lem:errorpropagation} is used to express \(\|T^{N_t-1}\We_{N_t}^\ast-\We_1^\ast\|_{p,\eta}\) as a function of the errors introduced by the approximate mappings.
Similar to the first chain of inequality in the proof of Lemma~\ref{lem:errorpropagation} applied at step $k=0$ and point $x_0$,
we obtain that%\\[-2em]
{\small%
 \begin{align*}
&\left|\left(\T^{N_t}\We^\ast_{N_t}\right)(x_0)- \left(\Te \We^\ast_{1}\right)(x_0)\right|\\& \leq \max_{a\in\M{A}}\int_{A\setminus K}  \left|\T^{N_t-1}\We^\ast_{N_t}(x_1)-\We^\ast_1(x_1)\right| \tx{x_1}{x_0,a}dx_1
%\\& \quad
+\left|\left(\T\We^\ast_1\right)(x_0)-\left(\Te \We^\ast_{1}\right)(x_0)\right|.
\end{align*}}
Let us now introduce a measure for the maximum concentration of the density function $\tx{x_1}{x_0,a}$ over $x_1 \in A \setminus K$,
for any $a \in \M{A}$,
defined relative to the density of the distribution $\eta$ in (\ref{eq:scaling_state}),
as $B_0=\sup_{x_1\in A\setminus K} \max_{a\in \M{A}}\frac{\tx{x_1}{x_0,a}}{\eta(x_1)}$.
%The value $B_0$ is defined relative to the uniform density distribution $\eta$.
Since $B_0\eta(x_1) \geq \tx{x_1}{x_0,a}$, it follows that%\\[-2em]
{\small%
 \begin{align*}
&\max_{a\in\M{A}}\int_{A\setminus K}  \left|\T^{N_t-1}\We^\ast_{N_t}(x_1)-\We^\ast_1(x_1)\right| \tx{x_1}{x_0,a}dx_1\leq B_0\int_{A\setminus K}  \left|\T^{N_t-1}\We^\ast_{N_t}(x_1)-\We^\ast_1(x_1)\right|  \eta(x_1) dx_1.
\end{align*}}
The last expression corresponds to a $1$-norm with respect to a probability measure $\eta$ over $A\setminus K$.
Exploiting the monotonicity of the $p$-norm with respect to a probability measure, a more general expression for the approximation error is obtained as
\begin{align*} 
\left|\left(\T^{N_t}\We^\ast_{N_t}\right)(x_0)- \Te \We^\ast_{1}(x_0)\right|\leq \left|\left(\Te \We^\ast_{1}\right)(x_0)-\left(\T\We_1^\ast\right)(x_0)\right|+B_0\left\|\left(\T^{N_t-1}\We_{N_t}^\ast\right)-\We_1^\ast\right\|_{p,\eta}.
\end{align*}%
%The value of $p$ is given.
The second term can be expressed as a function of the weighted $p$-norm of the approximations by applying Lemma \ref{lem:errorpropagation}.
This leads to the expression for an upper bound on the approximation error as
\begin{align*}\textstyle
\left|\rFVI-r_{x_0}^\ast(K,A)\right|\leq \left|\left(\Te \We^\ast_{1}\right)(x_0)-\left(\T\We_1^\ast\right)(x_0)\right|
%\\&\qquad
+B_0 \sum_{k=1}^{N_t-1}B^{\frac{k-1}{p}}\left\|\We_k^\ast-\T\We_{k+1}^\ast\right\|_{p,\eta}.
\end{align*}
From the above expression,
a sufficient condition the accuracy in \eqref{eq:acc} to hold is
\begin{align*}
\mathbf{P}\Big\{\left|\left(\Te \We^\ast_{1}\right)(x_0)-\left(\T\We_1^\ast\right)(x_0)\right|+B_0 \sum_{k=1}^{N_t-1}B^{\frac{k-1}{p}}\left\|\We_k^\ast-\T\We_{k+1}^\ast\right\|_{p,\eta}>\acc\Big\}\leq \Pacc.%\\[-4.5em]
\end{align*}%
\end{proof}

\section{Sample Complexities}
\label{App:samplecomplexity}
Given $\epsilon_{0,1,2}$ and $\alpha$, select $\delta_{0,1,2}>0$ such that $1-\alpha=\delta_0+(N_t-1)\delta_1+(N_t-1)\delta_2$,
and let us pick values for $N$,$M$, $M_0$ such that
\begin{align*}
\delta_0 \leq 2{|\M{A}|} e^{-2M_0(\epsilon_0)^2}, \quad
\delta_1 \leq 2{|\M{A}|N} e^{-2M(\epsilon_1)^2}, \quad
\delta_2 \leq4 e(d+1)\Big(\frac{32 e}{\epsilon_2^p} \Big)^d e^{-\frac{N\epsilon_2^{2p}}{128}}.
\end{align*}%\\[-1.5em]
Note that the first two inequalities are approximated with first order approximation for which we know that $1-(1-2e^{-2M_0(\epsilon_0)^2})^{|\M{A}|}\leq 2{|\M{A}|} e^{-2M_0(\epsilon_0)^2}$ and $1-(1-2e^{-2M(\epsilon_1)^2})^{|\M{A}|N}\leq 2{|\M{A}|N} e^{-2M(\epsilon_1)^2}$. The obtained integer values for $N$,$M$, $M_0$ are given as
\begin{equation*}
\left\{
\begin{array}{lll}
N&=&\Big{\lceil} 128\left(\ln(4e(d+1))+d \ln(32e) \right)\left(\frac{1}{\epsilon_2}\right)^{2p} +128dp\left(\frac{1}{\epsilon_2}\right)^{2p} \ln\left(\frac{1}{\epsilon_2}\right)+128\left(\frac{1}{\epsilon_2}\right)^{2p}\ln\left(\frac{1}{\delta_2}\right)
\Big{\rceil},\\
M&=&\Big{\lceil}\frac{1}{2}\left(\frac{1}{\epsilon_1}\right)^2\left(\ln(2|\M{A}|)+\ln(\frac{1}{\delta_1}) +\ln(N)\right)\Big{\rceil}, \\
M_0&=&\Big{\lceil}\frac{1}{2}\left(\frac{1}{\epsilon_0}\right)^2\left(\ln(2|\M{A}|)+\ln(\frac{1}{\delta_0})\right)\Big{\rceil},
\end{array}
\right. .
\end{equation*}
The use of the obtained $M,M_0,N$ in \eqref{eq:confidenceerrorbound} leads to a confidence of at least $\alpha$.

% !TEX root = ALearningApproach.tex

\section{Proof of Theorem \ref{thm:DataDependent}}\label{proof:emp_acc}
\begin{proof}
The proof of Theorem \ref{thm:DataDependent} is built observing that (a.) the single step error $\|\We^\ast_{k} -\T  \We^\ast_{k+1}\|_{1,\eta}$ is bounded by the sum of the expectations of (\ref{eq:emp_reg}) and (\ref{eq2:bias}); that (b.) the propagation of the single step errors gives a bound on the overall approximation error, see Theorem \ref{thm:MultiStep} --  hence the expected value of the  estimates, propagated over the time horizon, also gives a bound on the approximation error; and that (c.) the one-sided application of the Hoeffding's inequality provides a probabilistic upper bound on the deviation of the estimate from its mean, and therefore also bounds the approximation error probabilistically.  
\\
\textbf{Part (a.)}
%Define the expected value over random variable $x$ as $\Eo{x}{f(x)}$ for a given density distribution of $x$. 
\begin{align*}
\|\We^\ast_{k}-\T \We^\ast_{k+1}\|_{1,\eta}
&=\Eo{x}{\ \left|\We^\ast_{k}(x)-\T \We^\ast_{k+1}(x)\right|\ } % x\sim\eta 
\mbox{ with $\Eo{x}{f(x)}$ the mean of $f(x)$ for $x\sim \eta$.} %Definition of the norm (\ref{eq:objective}) 
%&=\mbox{[ Operator $\T$  (\ref{eq:T})  with simplified notation $V_{k+1}(y) =\ind_{K}(y)+\ind_{A\setminus K}(y) \We^\ast_{k+1}(y)$. ]}\\
%& =\Eo{x}{\left|\We^\ast_{k}(x)- \max_{a\in\M{A}}\Eo{y}{\ V_{k+1}(y)|x,a}\right|\ }\quad y\sim\Tx{\cdot}{x,a}, %\textmd{with} $V_{k+1}(y) =\ind_{K}(y)+\ind_{A\setminus K}(y) \We^\ast_{k+1}(y)$.%\label{eq:datadep1}
\end{align*}
Define a set of i.i.d. random variables $\vec{y}_1=[y^{a,1}_1,y^{a,2}_1,\ldots,y^{a,\tilde{M}}_1]$  drawn from the distribution $y^{a,j}_1\sim\Tx{\cdot}{x,a}$. Introduce  $\Eo{\vec{y}_1}{\max_{a\in\M{A}}\Te^a_1 \We^\ast_{k+1}(x) |x}$ as an auxiliary variable with $\Te_1^a$ the estimated operator as defined in (\ref{eq:empT1}) and computed over the  $\vec{y}_1$.
{\scriptsize \begin{align}
%&=\mbox{[ Introduce $\Eo{\vec{y}_1}{\max_{a\in\M{A}}\frac{1}{\tilde{M}} \sum_{j=1}^{\tilde{M}}V_{k+1}(y^{a,j}_1)|x}$ ]}\notag\\
& =\Eo{x}{\ \left|\We^\ast_{k}(x)-\Eo{\vec{y}_1}{\max_{a\in\M{A}}\Te^a_1 \We^\ast_{k+1}(x) |x} \right.\right.%\\&\qquad
+\left.\left.\Eo{\vec{y}_1}{\max_{a\in\M{A}}\Te^a_1 \We^\ast_{k+1}(x) |x}-
\T \We^\ast_{k+1}(x)\right|\ }\notag\\
&\leq \Eo{x}{\left|\We^\ast_{k}(x)-\Eo{\vec{y}_1}{\max_{a\in\M{A}}\Te^a_1 \We^\ast_{k+1}(x) |x}\right|} %\notag\\ &\qquad 
+ \Eo{x}{\left|\Eo{\vec{y}_1}{\max_{a\in\M{A}}\Te^a_1 \We^\ast_{k+1}(x) |x} - \T \We^\ast_{k+1}(x)\right|}  \notag\\
&\leq \underbrace{\mathbf{E}_{x,\vec{y_1}}\left[\left|\We^\ast_{k}(x)-\max_{a\in\M{A}}\Te^a_1 \We^\ast_{k+1}(x)\right|\right]}_{\textmd{\scriptsize{ [ Single step error ]}}} %\notag\\ &\qquad 
+\underbrace{\Eo{x}{\left|\Eo{\vec{y}_1}{\max_{a\in\M{A}}\Te^a_1 \We^\ast_{k+1}(x) |x} - \T \We^\ast_{k+1}(x)\right|}}_{  \textmd{\scriptsize{[ Bias term ]}}} .\notag
\end{align}}%
Observe that the \emph{single step error}, $\mathbf{E}_{x,\vec{y_1}}\left[\big|\We^\ast_{k}(x)-\max_{a\in\M{A}}\Te^a_1 \We^\ast_{k+1}(x)\big|\right]$ is equal to $\|\We^\ast_{k} -\Te  \We^\ast_{k+1}\|_{1,\eta}$ and $\mathbf{E}\|\We^\ast_{k} -\Te  \We^\ast_{k+1}\|_{1,\tilde{\eta}}$.
The \emph{bias term} gives the bias introduced by using an estimate of the operator and 
it can be rewritten as the expected value of (\ref{eq2:bias}). Note that $\max_{a\in\M{A}}\Eo{y}{ V_{k+1}(y)|x,a}$ is a function of $x,a$ and  $|\Eo{\vec{y}_1}{f(\vec{y})}|\leq\Eo{\vec{y}_1}{|f(\vec{y})|}$, thus it follows that
{\scriptsize\begin{align*}
%&\mbox{ [ Second term ]}\\
%&= \mathbf{E}_x\left[\left|\Eo{\vec{y}_1}{\left.\max_{a\in\M{A}}\frac{1}{\tilde{M}} \sum_{j=1}^{\tilde{M}}V_{k+1}(y^{a,j}_1)-
%\max_{a\in\M{A}}\Eo{y}{ V_{k+1}(y)|x,a}\right|x}\right|\right]\\
%&\mbox{ [ Note : $|\Eo{\vec{y}_1}{f(\vec{y})}|\leq\Eo{\vec{y}_1}{|f(\vec{y})|}$  ]}\\
 \textmd{\scriptsize{[ Bias term ]}}& \leq  \mathbf{E}_x\Eo{\vec{y}_1}{\left.\left|\max_{a\in\M{A}}\Te^a_1 \We^\ast_{k+1}(x)  -\T \We^\ast_{k+1}(x) \right|\ \right|x}
%\\&
 \leq  \mathbf{E}_x\Eo{\vec{y}_1}{\left.\max_{a\in\M{A}} \left|\Te^a_1 \We^\ast_{k+1}(x)-
\Eo{\vec{y}_2}{\Te^a_2 \We^\ast_{k+1}(x) |x}\right|\ \right|x}.
\end{align*}}%
The second inequality follows from introducing
 a secondary set of random variables $\vec{y}_2=[y^{a,1}_2,y^{a,2}_2,\ldots,y_2^{a,1\tilde{M}}]$  for which the elements are i.i.d. as $y^{a,1}_2\sim\Tx{\cdot}{x,a}$ and which are independent of $\vec{y}_1$. 
 Substituting $\T \We^\ast_{k+1}(x)$ with $    \max_{a\in\M{A}} \Eo{\vec{y}_2}{\Te^a_2 \We^\ast_{k+1}(x) |x,a}$ we have 
 {\scriptsize
\begin{align*}
%&\quad = \mathbf{E}_x\Eo{\vec{y}_1}{\max_{a\in\M{A}} \left|\Te^a_1 \We^\ast_{k+1}(x)-  \Eo{\vec{y}_2}{\Te^a_2 \We^\ast_{k+1}(x) |x,a} \right||x}\\
 \textmd{\scriptsize{[ Bias term ]}}&\leq \mathbf{E}_x\Eo{\vec{y}_1}{\max_{a\in\M{A}}\Eo{\vec{y}_2}{ \left|\Te^a_1 \We^\ast_{k+1}(x) -
 \Te^a_2 \We^\ast_{k+1}(x) \right||x,a,\vec{y}_1}|x}\\& =\mathbf{E}_{x,\vec{y}_1,\vec{y}_2}\left[\max_{a\in\M{A}} \left|\Te^a_1 \We^\ast_{k+1}(x)-
 \Te^a_2 \We^\ast_{k+1}(x)\right|\right] .\end{align*}}
The last equality is equal to the expected value of the estimated bias  $\mathbf{E} \big\|\max_{a\in\M{A}} \big|\Te^{a}_1\We^\ast_{k+1}- \Te^a_2\We^\ast_{k+1}\big|\big\|_{1,\tilde{\eta}}$. This proves statement (a.).
% The first term can be upper bounded as 
%\begin{align*}
%&\mbox{[First term]}\\
%&=\Eo{x}{\left|\We^\ast_{k}(x)-\Eo{\vec{y}_1}{\max_{a\in\M{A}}\frac{1}{\tilde{M}} \sum_{j=1}^{\tilde{M}}\left(\ind_{K}(y^{a,j}_1)+ \ind_{A\setminus K}(y^{a,j}_1)  \We^\ast_{k+1}(y^{a,j}_1)\right)|x}\right|} \\ 
%& \leq  \mathbf{E}_x{\Eo{\vec{y}_1}{\left.\left|\We^\ast_{k}(x)-\max_{a\in\M{A}}\frac{1}{\tilde{M}} \sum_{j=1}^{\tilde{M}}\left(\ind_{K}(y^{a,j}_1)  + \ind_{A\setminus K}(y^{a,j}_1)    \We^\ast_{k+1}(y^{a,j}_1)  \right)\right|\ \right|x}} \\ 
%&\mbox{[ Introduce 
%$\Te^{a}_1\We^\ast_{k+1}(x)$ as in  (\ref{eq:empT1}) ]}\\
%&
%=  \Eo{}{\left|\We^\ast_{k}(x)-\max_{a\in\M{A}}\Te^{a}_1\We^\ast_{k+1}(x)\right|}
%\end{align*}
%In conclusion, the error of an iteration can be upper bounded by the sum of the expectations of (\ref{eq:emp_reg}) and (\ref{eq:bias}),
%\begin{align}\label{app:step1}
%\|\We^\ast_{k}-\T \We^\ast_{k+1}\|_{1,\eta}
%&\leq \Eo{}{\left|\We^\ast_{k}(x)-\max_{a\in\M{A}}\Te^{a}_1\We^\ast_{k+1}(x)\right|}
% +\mathbf{E}\left[\max_{a\in\M{A}} \left|\Te^{a}_1 \We^\ast_{k+1}(x)-
%\Te^{a}_2\We^\ast_{k+1}(x)\right|\right].
%\end{align}

\smallskip

\textbf{Part (b.) $\&$ (c.)} 
Based on Theorem \ref{thm:MultiStep} $\rFVI$ has accuracy $\acc$ with probability $\Pacc$ if 
%The sum of the expectation of random variables is equal to the expectation of the sum of random variables, therefore the error of the FVI algorithm is bounded by 
%{\scriptsize \begin{align}
%&\left|\rFVI-r_{x_0}^{\ast}(K,A)\right|\leq \underbrace{|\We^\ast_0(x_0)-\T \We^\ast_1(x_0)|}_{\textmd{ [First term] }}
%+ \underbrace{B_0\sum_{k=1}^{N_t-1}B^{k-1} \Eo{}{\left|\We^\ast_{k}(x)-\max_{a\in\M{A}}\Te^{a}_1\We^\ast_{k+1}(x)\right|
% +\max_{a\in\M{A}} \left|\Te^{a}_1 \We^\ast_{k+1}(x)-
%\Te^{a}_2\We^\ast_{k+1}(x)\right|}}_{\textmd{ [Second term] }}
%\notag
%\end{align}}
{\scriptsize\begin{align}
& \mathbf P\Big\{  |\We^\ast_0(x_0)-\T \We^\ast_1(x_0)| 
+ B_0\sum_{k=1}^{N_t-1}B^{k-1} \Eo{}{\left|\We^\ast_{k}(x)-\max_{a\in\M{A}}\Te^{a}_1\We^\ast_{k+1}(x)\right|
 +\max_{a\in\M{A}} \left|\Te^{a}_1 \We^\ast_{k+1}(x)-
\Te^{a}_2\We^\ast_{k+1}(x)\right|} \geq \acc\Big\}<\Pacc.
\notag
\end{align}%
}%
Which holds under a union bounding argument if
{\scriptsize \begin{align}
& \p{  |\We^\ast_0(x_0)-\T \We^\ast_1(x_0)|
\geq \epsilon_0}<\delta_0\label{app:eq:first}\\
& \p{ \begin{array}{r}B_0\sum_{k=1}^{N_t-1}B^{k-1} \Eo{}{\left|\We^\ast_{k}(x)-\max_{a\in\M{A}}\Te^{a}_1\We^\ast_{k+1}(x)\right|
 +\max_{a\in\M{A}} \left|\Te^{a}_1 \We^\ast_{k+1}(x)-
\Te^{a}_2\We^\ast_{k+1}(x)\right|}\geq B_0 \epsilon \\+
B_0\sum_{k=1}^{N_t-1}B^{k-1}\left(
\left\|\We^\ast_{k} -\Te  \We^\ast_{k+1}\right\|_{1,\tilde{\eta}} 
 + \left\|\max_{a\in\M{A}} \left|\Te^{a}_1\We^\ast_{k+1}- \Te^a_2\We^\ast_{k+1}\right|\right\|_{1,\tilde{\eta}}
\right)
\end{array}}<e^{-2\frac{\Nh \epsilon^2}{L^2}}\label{app:eq:second}
\end{align}}
\noindent and  $\acc$ and $\Pacc$ are given as \eqref{eq:emp_acc} and \eqref{eq:emp_acc_conf}. 
\\*
The probabilistic bound \eqref {app:eq:first} follows from Lemma \ref{eq:empiricalnorm} for the estimation error of an empirical norm with accuracy  $\epsilon_0,\delta_0$ obtained for $p=1$, $M=M_0$ and $N=1$ as long as $0<2e^{-2M_0\epsilon_0^2}<1$.  
The probabilistic bound \eqref {app:eq:second} follows from a one-sided Hoeffding's inequality \citep{Hoeffding}  with random variable 
%\begin{align*}
%&B_0\sum_{k=1}^{N_t-1}B^{k-1} \Eo{}{\left|\We^\ast_{k}(x)-\max_{a\in\M{A}}\Te^{a}_1\We^\ast_{k+1}(x)\right|
% +\max_{a\in\M{A}} \left|\Te^{a}_1 \We^\ast_{k+1}(x)-
%\Te^{a}_2\We^\ast_{k+1}(x)\right|} \\
%&=B_0\Eo{}{\sum_{k=1}^{N_t-1}B^{k-1} \left|\We^\ast_{k}(x)-\max_{a\in\M{A}}\Te^{a}_1\We^\ast_{k+1}(x)\right|
% +\max_{a\in\M{A}} \left|\Te^{a}_1 \We^\ast_{k+1}(x)-
%\Te^{a}_2\We^\ast_{k+1}(x)\right|} 
%\end{align*}
%The above inequality is the expected value of a random variable,
\begin{align*}\textstyle\sum_{k=1}^{N_t-1}B^{k-1}\left( \left|\We^\ast_{k}(x)-\max_{a\in\M{A}}\Te^{a}_1\We^\ast_{k+1}(x)\right|
 +\max_{a\in\M{A}} \left|\Te^{a}_1 \We^\ast_{k+1}(x)-
\Te^{a}_2\We^\ast_{k+1}(x)\right|\right),\end{align*} 
 obtained from the combination of random variable $x \sim \eta$ and conditional random variables $\vec{y}_1$ and $\vec{y}_2$ and taking values in the range $[0,2\sum_{k=1}^{N_t-1}B^{k-1}]$.
Note that its estimated of interest over $\Nh$ samples %\[\frac{1}{\Nh}\sum_{i=1}^{\Nh}\sum_{k=1}^{N_t-1}B^{k-1}\left( \left|\We^\ast_{k}(x_i)-\max_{a\in\M{A}}\Te^{a}_1\We^\ast_{k+1}(x_i)\right|  +\max_{a\in\M{A}} \left|\Te^{a}_1 \We^\ast_{k+1}(x_i)- \Te^{a}_2\We^\ast_{k+1}(x_i)\right|\right),\]
can be rewritten in the form of (\ref{eq:emp_acc}),
$$\textstyle\sum_{k=1}^{N_t-1}B^{k-1}\Big(
\big\|\We^\ast_{k} -\Te  \We^\ast_{k+1}\big\|_{1,\tilde{\eta}}
+ \big\|\max_{a\in\M{A}} \left|\Te^{a}_1\We^\ast_{k+1}- \Te^a_2\We^\ast_{k+1}\right|\big\|_{1,\tilde{\eta}}\Big).$$
%By using an one-sided Hoeffding's inequality \citep{Hoeffding} the above estimate  can be used to determine an upper bound on the expected value.
  This concludes the proof of Theorem \ref{thm:DataDependent}.
%%%% Verification %%%%
\end{proof}

\section{Scaling factor for case study}\label{scalingfactor}

Compute $B$ as in \eqref{eq:scaling_B} using the given density distribution of the transitions \eqref{eq:densitytrans}, as
\begin{align*}
B&=\sup_{y\in A \setminus K}\int_{A\setminus K}\frac{1}{\sqrt{|\Sigma|(2\pi)^2}} \max_{a\in \M{A}}\bigg( \frac{\exp\big(-\frac{1}{2} \left(y- \mu\right)^T \Sigma^{-1} \left(y- \mu\right) \big)\eta(x)}{\eta(y)}\bigg)dx\allowdisplaybreaks \\
&= [ \mbox{ $\mu$ is a function of $a$ and $x$, and $\eta(\cdot)$ is constant over $A\setminus K$ }] \\
&=\sup_{y\in A \setminus K}\int_{A\setminus K}\frac{1}{\sqrt{|\Sigma|(2\pi)^2}} \max_{a\in \M{A}} \exp\left(-\frac{1}{2} \left(y- \mu\right)^T \Sigma^{-1} \left(y- \mu\right) \right) dx \allowdisplaybreaks\\
&=[\mbox{Suppose $\Ca$ is invertible, and define } \bar{\mu}(y,a)= \Ca^{-1}y-\Ca^{-1}\Cb a-\Ca^{-1}\Cc , \ \bar{\Sigma}= \Ca^{-1}\Sigma\Ca^{-T}]\allowdisplaybreaks \\
&=\sup_{y\in A \setminus K}\frac{1}{|\Ca|}\int_{A\setminus K}\frac{1}{\sqrt{|\bar{\Sigma}|(2\pi)^2}} \max_{a\in \M{A}} \exp\left(-\frac{1}{2} \left(x- \bar{\mu}(y,a)\right)^T \bar{\Sigma}^{-1}  \left(x- \bar{\mu}(y,a)\right) \right) dx \allowdisplaybreaks\\
&\leq\sup_{y\in A \setminus K}\frac{1}{|\Ca|}\int_{A\setminus K}\frac{1}{\sqrt{|\bar{\Sigma}|(2\pi)^2}} \sum_{a\in \M{A}}\left( \exp\left(-\frac{1}{2} \left(x- \bar{\mu}(y,a)\right)^T \bar{\Sigma}^{-1}  \left(x- \bar{\mu}(y,a)\right) \right)\right) dx \allowdisplaybreaks\\
&=\sup_{y\in A \setminus K}\frac{1}{|\Ca|}\left( \sum_{a\in \M{A}}\int_{A\setminus K}\frac{1}{\sqrt{|\bar{\Sigma}|(2\pi)^2}} \exp\left(-\frac{1}{2} \left(x- \bar{\mu}(y,a)\right)^T \bar{\Sigma}^{-1}  \left(x- \bar{\mu}(y,a)\right)\right) dx  \right).\\
 \end{align*}%\\[-1.4em]
The integral is rewritten as one over a scaled $2$-dimension multivariate Gaussian density distribution with mean $\bar{\mu}$ and covariance $\bar{\Sigma}$.
With this result, it can be deduced that $B$ is smaller than $\frac{1}{|\Ca|}|\A|$ as 
\begin{align} 
B&\leq\sup_{y\in A \setminus K}\frac{1}{|\Ca|}\Big( \textstyle\sum\limits_{a\in \M{A}}\int\limits_{\X}\frac{1}{\sqrt{|\bar{\Sigma}|(2\pi)^2}} \exp\left(-\frac{1}{2} \left(x- \bar{\mu}(y,a)\right)^T \bar{\Sigma}^{-1}  \left(x- \bar{\mu}(y,a)\right)\right) dx  \Big)\label{eq:app:compB}\\&\leq\sup\limits_{y\in A \setminus K}\frac{1}{|\Ca|}   \textstyle\sum\limits_{a\in \M{A}}1 =\frac{1}{|\Ca|}|\A|.\notag
\end{align}

\end{document}